\documentclass[11pt,journal,onecolumn]{IEEEtran}

\usepackage{amsmath,amssymb,mathtools}
\usepackage{braket}
\usepackage{cite}
\usepackage[hidelinks]{hyperref}
\usepackage{microtype}
\allowdisplaybreaks[2]

\newtheorem{theorem}{Theorem}
\newtheorem{proposition}[theorem]{Proposition}
\newtheorem{lemma}[theorem]{Lemma}
\newtheorem{corollary}[theorem]{Corollary}
\newtheorem{definition}[theorem]{Definition}

\newcommand{\spann}{\operatorname{span}}
\newcommand{\NS}{\mathrm{NS}}

\newcommand{\supp}{\operatorname{supp}}

\title{Multiplicative Semidefinite Programming Upper Bounds on Quantum No-Signalling-Assisted Zero-Error Capacity}
\author{Xin Li, Yinan Li
\thanks{Xin Li and Yinan Li are with the School of Artificial Intelligence,
    Wuhan University, Wuhan 430072, China. Corresponding author: Xin Li
    \mbox{(e-mail: 2026202220033@whu.edu.cn).}}}
\begin{document}
\maketitle

	\begin{abstract}
	We study zero-error communication and exact simulation of quantum channels assisted by quantum no-signalling correlations. For a channel with Kraus operator space \(K\), we introduce a semidefinite programming parameter \(\Gamma(K)\) that is exactly multiplicative under tensor products and upper-bounds
	\(\Theta_{\rm NS}(K)=2^{C_{0,\rm NS}(K)}\), where
	\(C_{0,\rm NS}(K)\) denotes the asymptotic quantum no-signalling-assisted
	zero-error capacity.
	Among known upper bounds on \(\Theta_{\rm NS}\) given by semidefinite programs, \(\Gamma\) is the tightest. We also prove that the previously introduced simulation lower bound \(\Sigma^{-}(K)\) is exactly multiplicative and satisfies
	\(\Theta_{\rm NS}(K)\leq\Gamma(K)\leq\Sigma^{-}(K)\leq\Gamma(K)^2\).
	Using this multiplicativity, we show that neither
	\(S_{0,\rm NS}(K)=\log_2\Sigma^{-}(K)\)
	nor its regularized extension
	\(S_{0,\rm NS}(K)=
	\lim_{n\to\infty}n^{-1}\log_2\Sigma^{-}(K^{\otimes n})\)
	holds in general.
	This answers an open question about the asymptotic simulation cost
	\(S_{0,\rm NS}(K)\).
	For classical--quantum channels, we obtain sharp bounds. If \(r\) is the effective output dimension, then \(\Theta_{\rm NS}\leq\Gamma\leq(1+\sqrt r)\Theta_{\rm NS}/2\) and \(\Gamma\leq\Sigma^{-}\leq(1+\sqrt r)\Gamma/2\). If \(k\) is the minimum rank of the output support projections, then \(\Sigma^{-}/\Theta_{\rm NS}\leq[(1+\sqrt r)/(1+\sqrt k)]^2\), and this bound is sharp for every \(1\leq k\leq r\). Thus, we obtain the sharp bound
	\(S_{0,\rm NS}-C_{0,\rm NS}\leq2\log_2[(1+\sqrt r)/(1+\sqrt k)]\).
\end{abstract}

\begin{IEEEkeywords}
	Zero-error communication, quantum channels, no-signalling correlations, asymptotic capacity, semidefinite programming.
\end{IEEEkeywords}

\section{Introduction}
\label{sec:introduction}
Zero-error information theory studies the amount of information that can be transmitted through a noisy channel when no decoding error is permitted. For a classical discrete memoryless channel with confusability graph \(G\), Shannon showed that the maximum number of messages that can be transmitted with zero error in a single use is the independence number \(\alpha(G)\). To describe the optimal rate attainable over arbitrarily many channel uses, Shannon introduced the asymptotic zero-error capacity, obtained from the regularized independence number under the strong graph product
\cite{Shannon1956,KornerOrlitsky1998}.
This regularization is generally difficult to compute directly. A fundamental problem in classical zero-error information theory is therefore to construct computable parameters that upper-bound the asymptotic zero-error capacity. Lov\'asz and Haemers introduced the upper bounds \(\vartheta(G)\)
and \(\mathcal H(G)\), respectively. These bounds are incomparable in general
\cite{Lovasz1979,Haemers1979,Haemers1981,GriblingLi2020}.

Shannon's zero-error communication problem also extends to quantum channels \cite{DuanShi2008,CubittChenHarrow2011,CubittSmith2012}. A quantum state is represented by a positive semidefinite operator
\(\rho\) with \(\operatorname{Tr}\rho=1\). A quantum channel $\mathcal N:\mathcal L(A)\rightarrow\mathcal L(B)$
is a completely positive and trace-preserving (CPTP) linear map. It admits a Kraus representation
\(
\mathcal N(\rho)
=
\sum_k E_k\rho E_k^\dagger,
\;
\sum_kE_k^\dagger E_k=I_{A}.
\)

To study zero-error communication over a quantum channel \(\mathcal N\),
Duan, Severini, and Winter introduced the noncommutative graph
\(S_{\mathcal N}:=\operatorname{span}\{E_i^\dagger E_j:i,j\}\),
which describes the confusability of its input states
\cite{DuanSeveriniWinter2013}.
They also studied entanglement-assisted zero-error classical communication
over quantum channels. We denote the corresponding asymptotic capacity
by \(C_{0,*}\) and set \(\Theta_*:=2^{C_{0,*}}\).
They further introduced the quantum Lov\'asz number
\(\widetilde{\vartheta}\) and proved that it gives a computable,
multiplicative upper bound on \(\Theta_*\). 
Noncommutative analogues of classical graph parameters have also
been used to derive bounds for zero-error communication
\cite{Stahlke2016,LevenePaulsenTodorov2018,GriblingLi2020,BorelandTodorovWinter2021}.
Another line of work concerns entanglement-assisted homomorphisms
of noncommutative graphs and their connections with non-local games
\cite{Stahlke2016,MustoReutterVerdon2018,BrannanHarrisTodorovTurowska2023}.
Li and Zuiddam characterized the entanglement-assisted zero-error
capacity of noncommutative graphs in terms of their asymptotic
spectrum \cite{LiZuiddam2021}.

Duan and Winter studied zero-error communication assisted by quantum
no-signalling (QNS) correlations, generalizing entanglement assistance
\cite{DuanWinter2016}. 

A quantum no-signalling correlation can be represented by a bipartite CPTP map
\(\mathcal R:\mathcal L(A_i\otimes B_i)\to\mathcal L(A_o\otimes B_o)\).
The map \(\mathcal R\) is said to be non-signalling from Bob to Alice if there exists a quantum channel
$
\mathcal R_A:\mathcal L(A_i)\to\mathcal L(A_o)
$
such that, for every
\(X_{A_iB_i}\),
\[
\operatorname{Tr}_{B_o}\mathcal R(X_{A_iB_i})
=
\mathcal R_A\bigl(\operatorname{Tr}_{B_i}X_{A_iB_i}\bigr).
\]
Similarly, if there exists a channel
\(\mathcal R_B:\mathcal L(B_i)\to\mathcal L(B_o)\)
such that
\[
\operatorname{Tr}_{A_o}\mathcal R(X_{A_iB_i})
=
\mathcal R_B\bigl(\operatorname{Tr}_{A_i}X_{A_iB_i}\bigr),
\]
then \(\mathcal R\) is non-signalling from Alice to Bob.
A map \(\mathcal R\) satisfying both no-signalling constraints is called a quantum no-signalling correlation
\cite{BeckmanGottesmanNielsenPreskill2001,
	EggelingSchlingemannWerner2002,
	PianiHorodeckiHorodeckiHorodecki2006,DuanWinter2016}.
The two constraints state, respectively, that Alice's reduced output depends only on Alice's reduced input and not on Bob's input, and that Bob's reduced output depends only on Bob's reduced input and not on Alice's input. Hence, changing either party's local input cannot alter the output statistics observable by the other party, and \(\mathcal R\) itself cannot serve as a communication channel between them.

In the QNS-assisted communication model, the given noisy quantum channel is denoted by
\(
\mathcal N:\mathcal L(A_o)\rightarrow\mathcal L(B_i).
\)
Alice encodes the message into \(A_i\), and the QNS operation
\(\mathcal R\) produces the channel input \(A_o\) on Alice's side.
The system \(A_o\) is transmitted through \(\mathcal N\) to \(B_i\), after which
\(\mathcal R\) produces the final output \(B_o\) on Bob's side.
This procedure yields an effective channel from \(A_i\) to \(B_o\).
If there exists a QNS operation \(\mathcal R\) such that this effective channel acts as a noiseless classical channel on an \(M\)-dimensional classical system, then one use of \(\mathcal N\) can transmit \(M\) classical messages without error with QNS assistance.
A more explicit construction of the communication model and its Choi-operator formulation can be found in
\cite{DuanWinter2016,LaiDuan2017}.

Any bipartite channel implemented by local operations and shared
entanglement satisfies these no-signalling conditions.
Entanglement-assisted communication is therefore a special case of
QNS-assisted communication.

For a channel \(\mathcal N:\mathcal L(A)\to\mathcal L(B)\)
with Kraus operators \(\{E_k\}\), let
\(K:=K(\mathcal N)=\operatorname{span}\{E_k:k\}\).
Its Choi operator is
\(J_{AB}^{\mathcal N}
:=\sum_k\lvert E_k\rangle\!\rangle\langle\!\langle E_k\rvert\),
where \(\lvert E_k\rangle\!\rangle\) is obtained by stacking the columns
of \(E_k\) in fixed bases, and
\(\langle\!\langle E_k\rvert\) is its Hermitian adjoint.
We write \(P_{AB}:=\supp(J_{AB}^{\mathcal N})\) for its support projection,
which depends only on \(K\), and set \(Q_{AB}:=I_{AB}-P_{AB}\).

Duan and Winter proved that the maximum number of classical messages
that can be transmitted without error in one channel use with QNS
assistance depends only on \(K\) and equals \(\lfloor\Upsilon(K)\rfloor\),
where \(\Upsilon(K)\) is given by the following semidefinite program
(SDP) \cite{DuanWinter2016}:
\begin{equation}
	\begin{aligned}
		\Upsilon(K)=\max\quad
		&\operatorname{Tr} S_A\\
		\mathrm{s.t.}\quad
		&S_A\succeq0,\\
		&0\preceq U_{AB}\preceq S_A\otimes I_B,\\
		&\operatorname{Tr}_AU_{AB}=I_B,\\
		&\operatorname{Tr}\!\left[
		P_{AB}
		\left(S_A\otimes I_B-U_{AB}\right)
		\right]=0.
	\end{aligned}
	\label{sdp:Upsilon}
\end{equation}

The parameter \(\Upsilon\) is supermultiplicative under tensor products, namely, \(\Upsilon(K_1\otimes K_2)\geq\Upsilon(K_1)\Upsilon(K_2)\) \cite{DuanWinter2016}.  
The asymptotic QNS-assisted zero-error capacity is given by
\begin{equation}
	C_{0,\NS}(K)
	=
	\lim_{n\to\infty}\frac1n\log_2\Upsilon(K^{\otimes n}),
	\quad
	\Theta_{\NS}(K):=2^{C_{0,\NS}(K)}.
	\label{DefofQNSCapacity}
\end{equation}

Thus the one-shot problem is given by a single semidefinite program, whereas the general asymptotic capacity is a regularized quantity for which no single-letter formula is known
\cite{DuanWinter2016}.
Duan and Wang subsequently introduced an activated one-shot quantity \(\widehat\Upsilon(K)\) given by the SDP
\cite{DuanWangActivated2015}
\begin{equation}
	\begin{aligned}
		\widehat\Upsilon(K)=\max\quad
		&\operatorname{Tr}S_A\\
		\mathrm{s.t.}\quad
		&S_A\succeq0,\\
		&0\preceq U_{AB}\preceq S_A\otimes I_B,\\
		&\operatorname{Tr}_AU_{AB}\preceq I_B,\\
		&\operatorname{Tr}\!\left[
		P_{AB}\left(S_A\otimes I_B-U_{AB}\right)
		\right]=0.
	\end{aligned}
	\label{sdp:WidehatUpsilonMax}
\end{equation}
Its dual SDP is \cite{DuanWangActivated2015}
\begin{equation}
	\begin{aligned}
		\widehat\Upsilon(K)=\min\quad
		&\operatorname{Tr}T_B\\
		\mathrm{s.t.}\quad
		&T_B\succeq0,\quad V_{AB}=V_{AB}^{\dagger},\\
		&\operatorname{Tr}_BV_{AB}\succeq I_A,\\
		&V_{AB}\preceq I_A\otimes T_B,\\
		&Q_{AB}V_{AB}Q_{AB}\preceq0.
	\end{aligned}
	\label{sdp:WidehatUpsilonMin}
\end{equation}
And strong duality holds.

They also proved that activation does not change the asymptotic
capacity \cite{DuanWangActivated2015}:
\begin{equation}
	\Theta_{\NS}(K)
	=
	\lim_{n\to\infty}
	\widehat\Upsilon(K^{\otimes n})^{1/n}
	=
	\sup_{n\geq1}
	\widehat\Upsilon(K^{\otimes n})^{1/n}.
	\label{eq:ThetaWidehatReg}
\end{equation}

Shannon's classical upper-bound problem has a natural counterpart here: to find a computable single-letter upper bound on \(\Theta_{\rm NS}\), preferably one that is exactly multiplicative under tensor products. Duan and Winter obtained a computable upper bound on
\(\Theta_{\rm NS}(K)\) in terms of the semidefinite covering number
\(\Sigma(K)\), defined by the following SDP
\cite{DuanWinter2016}:
\begin{equation}
	\begin{aligned}
		\Sigma(K)=\min\quad
		&\operatorname{Tr} T_B\\
		\mathrm{s.t.}\quad
		&F_{AB}\succeq0,\quad T_B\succeq0,\\
		&F_{AB}=P_{AB}F_{AB}P_{AB},\\
		&\operatorname{Tr}_BF_{AB}=I_A,\\
		&F_{AB}\preceq I_A\otimes T_B.
	\end{aligned}
	\label{sdp:Sigma}
\end{equation}

We introduce an exactly multiplicative SDP upper bound \(\Gamma(K)\)
and prove that \(\Gamma(K)\leq\Sigma(K)\) for every Kraus operator space \(K\). Among known general upper bounds on \(\Theta_{\NS}\) given by semidefinite programs, \(\Gamma\) is the tightest.

QNS-assisted exact simulation asks how much noiseless classical
communication is required to implement a prescribed quantum
channel exactly.
For a fixed channel \(\mathcal N\), Duan and Winter showed that
the minimum classical message alphabet size is
\(\lceil\Sigma(\mathcal N)\rceil\), where
\(\Sigma(\mathcal N)=
\min\{\operatorname{Tr}T_B:
T_B\succeq0,\,
J_{AB}^{\mathcal N}\preceq I_A\otimes T_B\}\)
\cite{DuanWinter2016}.
When only a Kraus operator space \(K\) is specified, the
semidefinite covering number is obtained by minimizing this
quantity over channels satisfying \(K(\mathcal N)\subseteq K\).

Since \(\Sigma\) is submultiplicative, its regularization
\begin{equation*}
	S_{0,\NS}(K)
	:=
	\lim_{n\to\infty}
	\frac1n
	\log_2\Sigma(K^{\otimes n})
	=
	\inf_{n\geq1}
	\frac1n
	\log_2\Sigma(K^{\otimes n})
\end{equation*}
exists and gives the asymptotic QNS-assisted exact simulation cost associated with the Kraus operator space, and it satisfies
\cite{DuanWinter2016}: 
\begin{equation}
	C_{0,\NS}(K)
	\leq
	S_{0,\NS}(K)
	\leq
	\log_2\Sigma(K).
	\label{C0nsleqS0nsleqSigma}
\end{equation}

The semidefinite covering number $\Sigma$ is not exactly multiplicative in general: Wang and Duan constructed a Kraus operator space satisfying \(\Sigma(K\otimes K)<\Sigma(K)^2\). They also introduced the supermultiplicative quantity \(\Sigma^{-}(K)\),
\begin{equation}
	\begin{aligned}
		\Sigma^{-}(K)=\max\quad &\operatorname{Tr}S_A\\
		\mathrm{s.t.}\quad
		&S_A\succeq0,\quad U_{AB}\succeq0,\\
		&\operatorname{Tr}_AU_{AB}=I_B,\\
		&P_{AB}(U_{AB}-S_A\otimes I_B)P_{AB}\succeq0,
	\end{aligned}
	\label{sdp:SigmaMinusWD}
\end{equation}
for which \(\log_2\Sigma^{-}(K)\) is a lower bound on \(S_{0,\NS}(K)\)
\cite{WangDuan2016}. They left open whether
\(S_{0,\NS}(K)=\log_2\Sigma^{-}(K)\)
holds for every Kraus operator space \(K\).

Using an equivalent SDP formulation, we prove that \(\Sigma^{-}\)
is exactly multiplicative under tensor products, strengthening
its previously known supermultiplicativity.
Thus regularization does not improve the simulation lower bound
\(\log_2\Sigma^{-}(K)\).
We construct a family of Kraus operator spaces for which
this bound is strictly smaller than \(S_{0,\NS}(K)\), answering
the above question in the negative.

We further prove
\(\Gamma(K)\leq\Sigma^{-}(K)\leq\Gamma(K)^2\).
Consequently, \(\Sigma^{-}\) is also a multiplicative upper bound
on \(\Theta_{\NS}\).
Together, these results give
\[
\widehat\Upsilon(K)
\leq
\Theta_{\NS}(K)
\leq
\Gamma(K)
\leq
\Sigma^{-}(K)
\leq
\Gamma(K)^2,
\quad
\Gamma(K)\leq\Sigma(K).
\]

Fix an orthonormal input basis
\(\{\lvert x\rangle:x\in\mathcal X\}\).
A classical--quantum channel has the form
\(
\mathcal N(\rho)
=
\sum_{x\in\mathcal X}
\langle x\rvert\rho\lvert x\rangle\,\rho_x,
\)
where
\(\rho_x\)
is a density operator on the output system \(B\).
Let \(P_x:=\supp(\rho_x)\). Then the Choi support projection is
\begin{equation}
	P_{AB}
	=
	\sum_{x\in\mathcal X}
	\lvert x\rangle\langle x\rvert_A
	\otimes P_x.
	\label{eq:CqSupport}
\end{equation}
For classical--quantum channels, the semidefinite packing number
\(A(K)\) is given by the following SDP \cite{DuanWinter2016}:
\begin{equation}
	\begin{aligned}
		A(K)=\max\quad
		&\sum_{x\in\mathcal X}s_x\\
		\mathrm{s.t.}\quad
		&s_x\geq0
		\quad\forall x,\\
		&\sum_{x\in\mathcal X}s_xP_x
		\preceq I_B.
	\end{aligned}
	\label{sdp:SemidefinitePackingPrimal}
\end{equation}
Its dual is
\begin{equation*}
	\begin{aligned}
		A(K)=\min\quad
		&\operatorname{Tr} T_B\\
		\mathrm{s.t.}\quad
		&T_B\succeq0,\\
		&\operatorname{Tr}(P_xT_B)\geq1
		\quad\forall x\in\mathcal X.
	\end{aligned}
\end{equation*}
The two SDPs satisfy strong duality.

Prior work gives
\(\Theta_{\NS}(K)=A(K)\) \cite{DuanWinter2016},
\(\widehat\Upsilon(K)=A(K)\) \cite{DuanWangActivated2015},
and
\(\Sigma^{-}(K)=\Sigma(K)\) \cite{WangDuan2016}.

For a classical--quantum channel \(x\mapsto\rho_x\), let \(K\)
denote its Kraus operator space and set
\(B':=\operatorname{span}\{\operatorname{ran}P_x:x\in\mathcal X\},\; r:=\dim B'\).
We prove the two sharp estimates
\[
A(K)\leq\Gamma(K)\leq\frac{1+\sqrt r}{2}A(K),
\quad
\Gamma(K)\leq\Sigma(K)\leq\frac{1+\sqrt r}{2}\Gamma(K).
\]
An explicit \(r\)-dimensional pure-state family attains equality in both upper bounds and therefore also attains
\(\Sigma/A=[(1+\sqrt r)/2]^2\). Moreover, this family provides a separation among \(\Theta_{\NS}\), \(\Gamma\), and \(\Sigma\).

The classical--quantum bounds admit a rank-sensitive refinement. Let
\(k:=\min_x\operatorname{rank}P_x\) be the minimum rank of these support projections. Then
\[
\frac{\Sigma(K)}{A(K)}
\leq
\left(\frac{1+\sqrt r}{1+\sqrt k}\right)^2,
\quad
\frac{\Sigma(K)}{\Gamma(K)}
\leq
\frac{1+\sqrt r}{1+\sqrt k}.
\]
The bound on \(\Sigma/A\) is sharp for every \(1\leq k\leq r\). We further obtain the sharp operational estimate
\[
0\leq S_{0,\NS}(K)-C_{0,\NS}(K)
\leq
2\log_2\frac{1+\sqrt r}{1+\sqrt k}.
\]

The paper is organized as follows. Section~II fixes our notation
and reviews Kraus operator spaces, noncommutative graphs, and
classical--quantum channels. Section~III introduces \(\Gamma\), derives its dual program, and proves its capacity upper-bound property and exact tensor multiplicativity. Section~IV relates \(\Gamma\) to \(\Sigma^{-}\), proves exact multiplicativity of \(\Sigma^{-}\) and the quadratic inequality \(\Sigma^{-}\leq\Gamma^2\),  summarizes these general bounds, and constructs a family for which
\(\log_2\Sigma^{-}(K)<S_{0,\NS}(K)\). Section~V treats the classical--quantum case and proves the sharp bounds
\(A\leq\Gamma\leq(1+\sqrt r)A/2\),
\(\Gamma\leq\Sigma\leq(1+\sqrt r)\Gamma/2\),
\(\Sigma/A\leq[(1+\sqrt r)/(1+\sqrt k)]^2\), and
\(S_{0,\NS}(K)-C_{0,\NS}(K)\leq
2\log_2[(1+\sqrt r)/(1+\sqrt k)]\). The duality arguments, classical--quantum reductions, and the comparison with the quantum Lov\'asz number are collected in the appendices.

\section{Preliminaries}
\label{sec:preliminaries}

This section fixes our notation and collects definitions and properties
used in the proofs. Unless stated otherwise, all Hilbert spaces are finite-dimensional and complex, and all logarithms are to base \(2\).

\subsection{Operator Notation, Choi Operators, and Kraus Operator Spaces}

For Hilbert spaces \(A\) and \(B\), let \(\mathcal L(A,B)\) denote the space of linear operators from \(A\) to \(B\), and write \(\mathcal L(A):=\mathcal L(A,A)\).
We write \(\operatorname{Herm}(A)\) and \(\operatorname{Pos}(A)\) for the vector space of Hermitian operators on \(A\) and the cone of positive semidefinite operators on \(A\), respectively.
For Hermitian operators \(X,Y\), the notation \(X\succeq Y\) means that \(X-Y\succeq0\).
The identity operator is denoted by \(I_A\), and the trace, rank, range, support projection, and operator norm of an operator are denoted by \(\operatorname{Tr}X\), \(\operatorname{rank}X\), \(\operatorname{ran}X\), \(\supp(X)\), and \(\lVert X\rVert_\infty\), respectively.
On a complex operator space, the Hilbert--Schmidt inner product is \(\langle X,Y\rangle:=\operatorname{Tr}(X^\dagger Y)\).
Partial traces are denoted by
\(\operatorname{Tr}_A\) and \(\operatorname{Tr}_B\).

Let \(A'\simeq A\), and fix the unnormalized maximally entangled vector
\(
\lvert\Phi\rangle_{AA'}
:=
\sum_i
\lvert i\rangle_{A}\otimes\lvert i\rangle_{A'}.
\)
For
\(E\in\mathcal L(A',B)\), define its vectorization by
\(
\lvert E\rangle\!\rangle_{AB}
:=
(I_{A}\otimes E)\lvert\Phi\rangle_{AA'}.
\)
This convention satisfies $\langle\!\langle E\vert F\rangle\!\rangle=\operatorname{Tr}(E^\dagger F).$

A quantum channel $\mathcal N:\mathcal L(A')\rightarrow\mathcal L(B)$
is a completely positive and trace-preserving linear map. It admits a Kraus representation
\(
\mathcal N(\rho)
=
\sum_k E_k\rho E_k^\dagger,
\;
\sum_kE_k^\dagger E_k=I_{A'}.
\)
The corresponding Choi operator is
\(J_{AB}^{\mathcal N} := (\operatorname{id}_{A}\otimes\mathcal N) \bigl( \lvert\Phi\rangle\langle\Phi\rvert_{AA'} \bigr) = \sum_k \lvert E_k\rangle\!\rangle \langle\!\langle E_k\rvert\).
Choi's theorem states that
\(
\mathcal N\ \text{is completely positive}\Leftrightarrow
J_{AB}^{\mathcal N}\succeq0,
\)
while trace preservation is equivalent to
\(
\operatorname{Tr}_BJ_{AB}^{\mathcal N}=I_{A}
\)
\cite{Choi1975}.

From this point onward, we write
\(\mathcal N:\mathcal L(A)\rightarrow\mathcal L(B)\) and
\(E\in\mathcal L(A,B)\)
in place of
\(\mathcal N:\mathcal L(A')\rightarrow\mathcal L(B)\) and
\(E\in\mathcal L(A',B)\), respectively, without ambiguity.

The Kraus operator space of a channel \(\mathcal N\) is defined by
\(
K(\mathcal N)
:=
\operatorname{span}\{E_k:k\}
\subseteq\mathcal L(A,B).
\)
This linear space is independent of the choice of Kraus representation. The Choi support projection of \(K\) is the orthogonal projection
\(P_{AB}\) onto
\(\operatorname{span}\{\lvert E\rangle\!\rangle_{AB}:E\in K\}\).
We set \(Q_{AB}:=I_{AB}-P_{AB}\).
When the underlying systems are clear, we write simply \(P\) and \(Q\).

A subspace \(K\subseteq\mathcal L(A,B)\) is called a Kraus operator space if there exists a CPTP channel \(\mathcal N\) such that \(K(\mathcal N)=K\).
In zero-error information theory, it is also called a noncommutative bipartite graph
\cite{DuanSeveriniWinter2016,DuanWinter2016}.
Consequently, if \(P_{AB}\) is the support projection of \(K\), then there exists a channel Choi operator
\(J_{AB}\) satisfying
\(
J_{AB}\succeq0,\;
\supp(J_{AB})=P_{AB},\;
\operatorname{Tr}_BJ_{AB}=I_A.
\)

For two Kraus operator spaces \(K_i\subseteq\mathcal L(A_i,B_i)\), \(i=1,2\), their tensor product is
\(
K_1\otimes K_2
\subseteq
\mathcal L(A_1\otimes A_2,B_1\otimes B_2).
\)
Define the standard reordering unitary
\(\Pi_{12}: A_1\otimes B_1\otimes A_2\otimes B_2 \to A_1\otimes A_2\otimes B_1\otimes B_2\)
by
\[
\Pi_{12}
\bigl(
\lvert a_1\rangle
\otimes\lvert b_1\rangle
\otimes\lvert a_2\rangle
\otimes\lvert b_2\rangle
\bigr)
=
\lvert a_1\rangle
\otimes\lvert a_2\rangle
\otimes\lvert b_1\rangle
\otimes\lvert b_2\rangle.
\]
Once the two Kraus operator spaces under consideration have been fixed, we abbreviate
\(\Pi_{12}\) as \(\Pi\) whenever no confusion can arise.
If \(P_i\) is the Choi support projection of \(K_i\), then
\(P_{12} = \Pi_{12} (P_1\otimes P_2) \Pi_{12}^\dagger\)
is the support projection of \(K_1\otimes K_2\).

\subsection{Noncommutative Graphs and Entanglement-Assisted Zero-Error Capacity}

The noncommutative graph induced by a Kraus operator space \(K\) is
\(S(K) := K^\dagger K = \operatorname{span} \{E^\dagger F:E,F\in K\} \subseteq\mathcal L(A)\).
It is a self-adjoint operator space containing the identity, namely, an operator system.
Suppose that a quantum channel \(\mathcal N\) satisfies
\(K(\mathcal N)=K\). For unit vectors
\(\lvert\psi\rangle,\lvert\varphi\rangle\in A\), the output states
\(\mathcal N(\lvert\psi\rangle\langle\psi\rvert)\)
and
\(\mathcal N(\lvert\varphi\rangle\langle\varphi\rvert)\)
are orthogonal if and only if
$
\langle\psi\vert X\vert\varphi\rangle=0,
\;
\forall\,X\in S(K).
$
Thus, the noncommutative graph determines the confusability relation among input states, whereas the Kraus operator space additionally determines the Choi support subspace in \(A\otimes B\)
\cite{DuanSeveriniWinter2013,DuanSeveriniWinter2016}.

In general, the following inequality holds \cite[(11)]{WangDuan2018}: \(C_{0,*}(S(K))\leq C_{0,\NS}(K)\).
In the nonlogarithmic notation used here, this relation is equivalent to
\begin{equation}
	\Theta_*(S(K))
	\leq
	\Theta_{\NS}(K).
	\label{eq:ThetaStarThetaNS}
\end{equation}

\subsection{Classical--Quantum Kraus Operator Spaces}

For a classical--quantum channel \(x\mapsto\rho_x\), write
\(P_x:=\supp(\rho_x)\). Its Kraus operator space is determined by
the projections \(\{P_x\}_{x\in\mathcal X}\), through the Choi support
projection \eqref{eq:CqSupport}.
Since every output state is supported on
\(B':=\operatorname{span}\{\operatorname{ran}P_x:x\in\mathcal X\}\),
we may restrict the output space to \(B'\).
Throughout the classical--quantum analysis, we write
\(r:=\dim B'\) and \(k:=\min_x\operatorname{rank}P_x\), so that
\(1\leq k\leq r\).
\section{An Exactly Multiplicative Single-Letter Upper Bound \texorpdfstring{$\Gamma$}{Gamma}}
\label{sec:Gamma}

\subsection{Construction and Basic Reductions}

The minimization formulation \eqref{sdp:WidehatUpsilonMin} of \(\widehat\Upsilon\) contains the constraint \(QVQ\preceq0\), which is not preserved by tensor products. A tensor-stable substitute is obtained by writing
\(V=(X+X^\dagger)/2\) and imposing \(PX=X\), which forces \(QVQ=0\). Likewise, writing \(T=(T_0+T_1)/2\) and imposing a \(2\times2\) block positive semidefinite constraint on \(X,T_0,T_1\) ensures \(V\preceq I_A\otimes T\). Together with \(\operatorname{Tr}_BX\succeq I_A\), these tensor-stable conditions define the following SDP.

\begin{definition}\label{def:Gamma}
	For every Kraus operator space \(K\), let \(\Gamma(K)\) be the optimal value of the following SDP:
	\begin{equation}
		\begin{aligned}
			\Gamma(K)=\min\quad &\frac12\operatorname{Tr}(T_{0,B}+T_{1,B})\\
			\mathrm{s.t.}\quad
			&PX_{AB}=X_{AB},\\
			&\operatorname{Tr}_BX_{AB}\succeq I_A,\\
			&\begin{pmatrix}
				I_A\otimes T_{0,B}&X_{AB}\\
				X_{AB}^{\dagger}&I_A\otimes T_{1,B}
			\end{pmatrix}\succeq0.
		\end{aligned}
		\label{sdp:Gamma}
	\end{equation}
	Here \(X_{AB}\in\mathcal L(A\otimes B)\) is an arbitrary complex operator satisfying the constraints.
\end{definition}

The positivity of the two diagonal principal blocks implies \(T_{0,B},T_{1,B}\succeq0\).

\begin{lemma}\label{lem:GammaNormalize}
	For every Kraus operator space \(K\), replacing the inequality constraint \(\operatorname{Tr}_BX_{AB}\succeq I_A\) in \eqref{sdp:Gamma} by \(\operatorname{Tr}_BX_{AB}=I_A\) does not change the optimal value. Equivalently, \(\Gamma\) has the following SDP formulation:
	\begin{equation}
		\begin{aligned}
			\Gamma(K)=\min\quad &\frac12\operatorname{Tr}(T_{0,B}+T_{1,B})\\
			\mathrm{s.t.}\quad
			&PX_{AB}=X_{AB},\\
			&\operatorname{Tr}_BX_{AB}= I_A,\\
			&\begin{pmatrix}
				I_A\otimes T_{0,B}&X_{AB}\\
				X_{AB}^{\dagger}&I_A\otimes T_{1,B}
			\end{pmatrix}\succeq0.
		\end{aligned}
		\label{sdp:GammaEq}
	\end{equation}
\end{lemma}

\begin{IEEEproof}
	Let \((X,T_0,T_1)\) be feasible for \eqref{sdp:Gamma}, and set
	\(H_A:=\operatorname{Tr}_BX\succeq I_A\),
	\(C_{AB}:=H_A^{-1}\otimes I_B\), and \(\widetilde X:=XC_{AB}\).
	Since \(H_A\succeq I_A\), the operator \(H_A\) is invertible and \(0\prec C_{AB}\preceq I_{AB}\). The support constraint is preserved, since \(P\widetilde X=PX C_{AB}=XC_{AB}=\widetilde X\). The partial trace satisfies \(\operatorname{Tr}_B\widetilde X=(\operatorname{Tr}_BX)H_A^{-1}=I_A\). Let \(D_0:=I_A\otimes T_0\) and \(D_1:=I_A\otimes T_1\).
	A congruence transformation of the original block matrix by \(\operatorname{diag}(I_{AB},C_{AB})\) gives
	\begin{equation}
		\begin{pmatrix}
			D_0&XC_{AB}\\
			C_{AB}X^{\dagger}&C_{AB}D_1C_{AB}
		\end{pmatrix}\succeq0.
		\label{lemma11proofeq1}
	\end{equation}
	Because \(D_1\) commutes with \(C_{AB}\) and \(0\preceq C_{AB}\preceq I_{AB}\),
	\(D_1-C_{AB}D_1C_{AB} =D_1^{\frac12}(I_{AB}-C_{AB}^{2})D_1^{\frac12}\succeq0\).
	Adding the positive semidefinite operator \(D_1-C_{AB}D_1C_{AB}\) to the lower-right block of \eqref{lemma11proofeq1} yields
	\(
	\begin{psmallmatrix}
		D_0&\widetilde X\\
		\widetilde X^{\dagger}&D_1
	\end{psmallmatrix}\succeq0.
	\)
	
	Therefore, \((\widetilde X,T_0,T_1)\) is a feasible solution of \eqref{sdp:GammaEq} with the same objective value. The reverse inclusion of feasible sets holds directly, so the SDPs \eqref{sdp:Gamma} and \eqref{sdp:GammaEq} are equivalent.
\end{IEEEproof}

For convenience, we refer below to the three constraints in \eqref{sdp:Gamma} and \eqref{sdp:GammaEq} as the support constraint, the partial-trace constraint, and the block positive semidefinite constraint, respectively.

\begin{lemma}\label{lem:GammaBalance}
	For every Kraus operator space \(K\), the infimum in \eqref{sdp:Gamma} and \eqref{sdp:GammaEq} may be restricted to feasible solutions satisfying \(\operatorname{Tr} T_{0,B}=\operatorname{Tr} T_{1,B}\).
\end{lemma}

\begin{IEEEproof}
	If \((X,T_0,T_1)\) is feasible, then for every \(\lambda>0\),
	\[
	\begin{pmatrix}
		I_A\otimes \lambda T_{0}&X\\
		X^{\dagger}&I_A\otimes \lambda^{-1}T_{1}
	\end{pmatrix}
	=
	\begin{pmatrix}
		\lambda^{1/2}I_{AB} & 0\\
		0 & \lambda^{-1/2}I_{AB}
	\end{pmatrix}
	\begin{pmatrix}
		I_A\otimes T_{0}&X\\
		X^{\dagger}&I_A\otimes T_{1}
	\end{pmatrix}
	\begin{pmatrix}
		\lambda^{1/2}I_{AB} & 0\\
		0 & \lambda^{-1/2}I_{AB}
	\end{pmatrix}
	\succeq0.
	\]
	Hence \((X,\lambda T_0,\lambda^{-1}T_1)\) remains feasible.
	No feasible solution can satisfy \(\operatorname{Tr} T_i=0\). Indeed, if, without loss of generality, \(\operatorname{Tr} T_0=0\), then the block positive semidefinite constraint implies \(T_0=0_B\) and \(X=0_{AB}\), contradicting \(\operatorname{Tr}_BX_{AB}\succeq I_A\). Thus \(\operatorname{Tr} T_0,\operatorname{Tr} T_1>0\) for every feasible solution.
	
	As a function of \(\lambda\), the objective value \(\frac12(\lambda\operatorname{Tr} T_0+\lambda^{-1}\operatorname{Tr} T_1)\) is minimized when \(\operatorname{Tr}(\lambda T_0)=\operatorname{Tr}(\lambda^{-1}T_1)\), which proves the lemma.
\end{IEEEproof}

\subsection{One-Shot Comparison with \texorpdfstring{$\widehat\Upsilon$}{Upsilon-hat}}

\begin{proposition}\label{prop:WidehatGamma}
	For every Kraus operator space \(K\),
	\[
	\widehat\Upsilon(K)\leq\Gamma(K).
	\]
\end{proposition}

\begin{IEEEproof}
	Let \((X,T_0,T_1)\) be feasible for \eqref{sdp:Gamma}, and define
	\(V:=(X+X^{\dagger})/2\) and \(T:=(T_0+T_1)/2\).
	The support constraint \(PX=X\) implies \(QX=0\) and \(X^{\dagger}Q=0\), where \(Q=I-P\). Hence \(QVQ=0\). Moreover, since \(\operatorname{Tr}_BX\succeq I_A\), the operator \(\operatorname{Tr}_BX\) is Hermitian, and hence \(\operatorname{Tr}_BV=\operatorname{Tr}_BX\succeq I_A\).
	
	Let \(M\) denote the block matrix in \eqref{sdp:Gamma}, and define the column operator
	\(W:=\begin{psmallmatrix}I_{AB}\\-I_{AB}\end{psmallmatrix}\).
	Since \(M\succeq0\),
	\[
	0\preceq W^{\dagger}MW
	=I_A\otimes(T_0+T_1)-X-X^{\dagger}.
	\]
	Thus \(V\preceq I_A\otimes T\). Consequently, \((V,T)\) is feasible for \eqref{sdp:WidehatUpsilonMin}. Its objective value is
	\(\operatorname{Tr}T=\frac12\operatorname{Tr}(T_0+T_1)\). Taking the infimum over feasible \((X,T_0,T_1)\) proves the claim.
\end{IEEEproof}

\subsection{Dual Program and Exact Tensor Multiplicativity}

\begin{proposition}\label{prop:GammaDual}
	The dual of \eqref{sdp:GammaEq} is
	\begin{equation}
		\begin{aligned}
			\Gamma(K)=\max\quad &\operatorname{Re}\operatorname{Tr} H_A\\
			\mathrm{s.t.}\quad
			&\begin{pmatrix}Y_{0,AB}&Z_{AB}\\Z_{AB}^{\dagger}&Y_{1,AB}\end{pmatrix}\succeq0,\\
			&\operatorname{Tr}_AY_{0,AB}=\operatorname{Tr}_AY_{1,AB}=I_B,\\
			&PZ_{AB}=P(H_A\otimes I_B),
		\end{aligned}
		\label{sdp:GammaDual}
	\end{equation}
	where \(H_A\) and \(Z_{AB}\) are arbitrary complex operators. Moreover, strong duality holds.
\end{proposition}

The derivation of the dual and the verification of strong duality are given in Appendix~\ref{app:GammaDual}.

\begin{theorem}\label{thm:GammaMult}
	For any Kraus operator spaces \(K_1,K_2\),
	\[
	\Gamma(K_1\otimes K_2)=\Gamma(K_1)\Gamma(K_2).
	\]
\end{theorem}

\begin{IEEEproof}
	For \(i=1,2\), let \((X_i,T_{0,i},T_{1,i})\) be balanced primal-feasible solutions for \(\Gamma(K_i)\) in \eqref{sdp:Gamma}, so that \(\operatorname{Tr} T_{0,i}=\operatorname{Tr} T_{1,i}\). Set
	\(X_{12}=\Pi_{12}(X_1\otimes X_2)\Pi_{12}^{\dagger},\; T_{0,12}=T_{0,1}\otimes T_{0,2},\; T_{1,12}=T_{1,1}\otimes T_{1,2}\).
	We first verify that this triple is primal feasible for \(\Gamma(K_1\otimes K_2)\) in \eqref{sdp:Gamma}.
	Let \(R_{ji}:=I_{A_i}\otimes T_{j,i}\) for \(i=1,2\) and \(j=0,1\), and set
	\(M_i:= \begin{psmallmatrix} R_{0i} & X_i\\ X_i^{\dagger} & R_{1i} \end{psmallmatrix} \succeq 0,\; i=1,2\).
	Then
	\begin{equation}
		M_1\otimes M_2\succeq0.
		\label{eq:Gammasubmulti1}
	\end{equation}
	Moreover, \(M_1\otimes M_2\) is unitarily equivalent to
	\[
	\begin{pmatrix}
		R_{01}\otimes R_{02}
		&
		R_{01}\otimes X_2
		&
		X_1\otimes R_{02}
		&
		X_1\otimes X_2
		\\
		R_{01}\otimes X_2^\dagger
		&
		R_{01}\otimes R_{12}
		&
		X_1\otimes X_2^\dagger
		&
		X_1\otimes R_{12}
		\\
		X_1^\dagger\otimes R_{02}
		&
		X_1^\dagger\otimes X_2
		&
		R_{11}\otimes R_{02}
		&
		R_{11}\otimes X_2
		\\
		X_1^\dagger\otimes X_2^\dagger
		&
		X_1^\dagger\otimes R_{12}
		&
		R_{11}\otimes X_2^\dagger
		&
		R_{11}\otimes R_{12}
	\end{pmatrix}
	\succeq0.
	\]
	Taking the corresponding principal submatrix gives
	\[
	\begin{pmatrix}
		R_{01}\otimes R_{02}
		&X_1\otimes X_2\\
		X_1^{\dagger}\otimes X_2^{\dagger}
		&R_{11}\otimes R_{12}
	\end{pmatrix}\succeq0.
	\]
	Consequently,
	\[
	\begin{pmatrix}
		\Pi&0\\0&\Pi
	\end{pmatrix}
	\begin{pmatrix}
		R_{01}\otimes R_{02}
		&X_1\otimes X_2\\
		X_1^{\dagger}\otimes X_2^{\dagger}
		&R_{11}\otimes R_{12}
	\end{pmatrix}
	\begin{pmatrix}
		\Pi&0\\0&\Pi
	\end{pmatrix}^\dagger
	=
	\begin{pmatrix}
		I_{A_1A_2}
		\otimes
		T_{0,12}
		&
		X_{12}
		\\
		X_{12}^\dagger
		&
		I_{A_1A_2}
		\otimes
		T_{1,12}
	\end{pmatrix}
	\succeq0.
	\]
	Thus the block positive semidefinite constraint holds. The support and partial-trace constraints follow by direct verification. The balance conditions \(\operatorname{Tr} T_{0,i}=\operatorname{Tr} T_{1,i}\) imply
	\[
	\frac12\operatorname{Tr}(T_{0,1}\otimes T_{0,2}+T_{1,1}\otimes T_{1,2})
	=\left[\frac12\operatorname{Tr}(T_{0,1}+T_{1,1})\right]
	\left[\frac12\operatorname{Tr}(T_{0,2}+T_{1,2})\right].
	\]
	Therefore, by Lemma~\ref{lem:GammaBalance}, taking the infimum over balanced feasible solutions \((X_i,T_{0,i},T_{1,i})\) gives
	\[
	\Gamma(K_1\otimes K_2)\leq\Gamma(K_1)\Gamma(K_2).
	\]
	
	For \(i=1,2\), let
	\((H_i,Y_{0,i},Y_{1,i},Z_i)\)
	be any dual-feasible solution for the dual program of \(\Gamma(K_i)\).
	Since the replacements \(H_i\mapsto e^{\mathrm{i}\phi}H_i\) and
	\(Z_i\mapsto e^{\mathrm{i}\phi}Z_i\)
	preserve dual feasibility, we may assume without loss of generality that
	\(\operatorname{Tr}H_i\) is real and nonnegative.
	Set
	\(H_{12}=H_1\otimes H_2,\; Z_{12}=\Pi(Z_1\otimes Z_2)\Pi^\dagger,\; Y_{0,12}=\Pi(Y_{0,1}\otimes Y_{0,2})\Pi^\dagger,\; Y_{1,12}=\Pi(Y_{1,1}\otimes Y_{1,2})\Pi^\dagger\).
	A direct calculation verifies that
	\(
	\operatorname{Tr}_{A_1A_2}Y_{j,12}
	=
	I_{B_1B_2},
	\ j=0,1,
	\)
	and
	\(
	P_{12}Z_{12}
	=
	P_{12}(H_{12}\otimes I_{B_1B_2})
	\).
	
	Because
	\(\begin{psmallmatrix} Y_{0,i} & Z_{i}\\ Z_{i}^{\dagger} & Y_{1,i} \end{psmallmatrix} \succeq0,\; i=1,2\),
	applying the same principal-submatrix argument as in \eqref{eq:Gammasubmulti1} to their tensor product yields
	\[
	\begin{pmatrix}
		Y_{0,1}\otimes Y_{0,2} & Z_{1}\otimes Z_{2}\\
		Z_{1}^{\dagger}\otimes Z_{2}^{\dagger} & Y_{1,1}\otimes Y_{1,2}
	\end{pmatrix}\succeq0.
	\]
	Therefore,
	\[
	\begin{pmatrix}
		Y_{0,12} & Z_{12}\\
		Z_{12}^{\dagger} & Y_{1,12}
	\end{pmatrix}
	=
	\begin{pmatrix}
		{\Pi}&0
		\\
		0&{\Pi}
	\end{pmatrix}
	\begin{pmatrix}
		Y_{0,1}\otimes Y_{0,2} & Z_{1}\otimes Z_{2}\\
		Z_{1}^{\dagger}\otimes Z_{2}^{\dagger} & Y_{1,1}\otimes Y_{1,2}
	\end{pmatrix}\begin{pmatrix}
		{\Pi}&0
		\\
		0&{\Pi}
	\end{pmatrix}^\dagger
	\succeq0.
	\]
	Hence the block positive semidefinite constraint is also satisfied. Thus \((H_{12},Y_{0,12},Y_{1,12},Z_{12})\) is dual feasible for the dual program of \(\Gamma(K_1\otimes K_2)\) in \eqref{sdp:GammaDual}, with objective value
	\(
	\operatorname{Tr}(H_1\otimes H_2)
	=
	\operatorname{Tr}(H_1)\operatorname{Tr}(H_2)
	\).
	Taking the supremum over feasible \((H_i,Y_{0,i},Y_{1,i},Z_i)\) and using strong duality gives
	\[
	\Gamma(K_1\otimes K_2)
	\geq
	\Gamma(K_1)\Gamma(K_2).
	\]
	Combining the two directions proves multiplicativity.
\end{IEEEproof}

\begin{corollary}\label{cor:ThetaGamma}
	For every Kraus operator space \(K\),
	\[
	\Theta_{\NS}(K)\leq\Gamma(K).
	\]
\end{corollary}

\begin{IEEEproof}
	By Proposition~\ref{prop:WidehatGamma} and Theorem~\ref{thm:GammaMult}, for every positive integer \(n\),
	\[
	\widehat\Upsilon(K^{\otimes n})
	\leq
	\Gamma(K^{\otimes n})
	=
	\Gamma(K)^n.
	\]
	Taking the \(n\)th root and using \eqref{eq:ThetaWidehatReg} gives the result.
\end{IEEEproof}

\section{Relation to Exact Simulation}
\label{sec:simulationRelation}

\subsection{An Equivalent Primal Form of \texorpdfstring{$\Sigma^{-}$}{Sigma-minus} and the Bound \texorpdfstring{$\Gamma\leq\Sigma^{-}$}{Gamma <= Sigma-minus}}
For comparison with \(\Gamma\), we derive the following minimization
formulation of \(\Sigma^{-}(K)\), equivalent to
\eqref{sdp:SigmaMinusWD}.

\begin{proposition}\label{prop:SigmaMinusPrimal}
	For every Kraus operator space \(K\), \(\Sigma^{-}(K)\) equals
	\begin{equation}
		\begin{aligned}
			\min\quad &\operatorname{Tr}T_B\\
			\mathrm{s.t.}\quad
			&V_{AB}\succeq0,\quad V_{AB}=PV_{AB}P,\\
			&\operatorname{Tr}_BV_{AB}\succeq I_A,\\
			&V_{AB}\preceq I_A\otimes T_B,\quad T_B\succeq0.
		\end{aligned}
		\label{sdp:SigmaMinusPrimal}
	\end{equation}
	Its dual is
	\begin{equation}
		\begin{aligned}
			\max\quad &\operatorname{Tr}R_A\\
			\mathrm{s.t.}\quad
			&R_A\succeq0,\quad X_{AB}\succeq0,\\
			&\operatorname{Tr}_AX_{AB}\preceq I_B,\\
			&P(X_{AB}-R_A\otimes I_B)P\succeq0,
		\end{aligned}
		\label{sdp:SigmaMinusDual}
	\end{equation}
	and strong duality holds.
\end{proposition}

\begin{IEEEproof}
	The Lagrange dual of \eqref{sdp:SigmaMinusPrimal} is \eqref{sdp:SigmaMinusDual}; the derivation and the verification of the relative Slater condition are given in Appendix~\ref{app:SigmaMinusDual}.
	
	Every feasible solution \((S_A,U_{AB})\) of \eqref{sdp:SigmaMinusWD} is feasible for \eqref{sdp:SigmaMinusDual} after setting \(R_A=S_A\) and \(X_{AB}=U_{AB}\). Conversely, let \((R_A,X_{AB})\) be feasible for \eqref{sdp:SigmaMinusDual}, and define
	\(D_B:=I_B-\operatorname{Tr}_AX_{AB}\succeq0\).
	For any density operator \(\rho_A\), the operator
	\(X'_{AB}:=X_{AB}+\rho_A\otimes D_B\) is positive semidefinite, satisfies \(\operatorname{Tr}_AX'_{AB}=I_B\), and obeys
	\(P(X'_{AB}-R_A\otimes I_B)P\succeq0\). The objective value is unchanged. Thus the two maximization programs have the same optimal value, and the conclusion follows from strong duality for \eqref{sdp:SigmaMinusPrimal}.
\end{IEEEproof}

Every feasible solution of \eqref{sdp:Sigma} is feasible for \eqref{sdp:SigmaMinusPrimal}. Hence
\(\Sigma^{-}(K)\leq\Sigma(K)\)
for every Kraus operator space \(K\).

\begin{proposition}\label{prop:GammaSigmaMinus}
	For every Kraus operator space \(K\),
	\[
	\Gamma(K)\leq\Sigma^{-}(K).
	\]
\end{proposition}

\begin{IEEEproof}
	Let \((V,T)\) be any feasible solution of \eqref{sdp:SigmaMinusPrimal}, and write \(R:=I_A\otimes T\). Then
	\(V\succeq0\), \(PV=V\), \(\operatorname{Tr}_BV\succeq I_A\), and \(0\preceq V\preceq R\). Set \(X:=V\) and \(T_0:=T_1:=T\) in the primal SDP \eqref{sdp:Gamma} for \(\Gamma\). The support and partial-trace constraints are immediate.
	
	For the block positive semidefinite constraint, consider
	\(U:=\frac1{\sqrt2} \begin{psmallmatrix} I_{AB}&I_{AB}\\ I_{AB}&-I_{AB} \end{psmallmatrix}\).
	A direct calculation gives
	\[
	U^{\dagger}
	\begin{pmatrix}
		R&V\\V&R
	\end{pmatrix}
	U
	=
	\begin{pmatrix}
		R+V&0\\0&R-V
	\end{pmatrix}
	\succeq0.
	\]
	Thus \((V,T,T)\) is feasible for \(\Gamma(K)\), with objective value \(\operatorname{Tr}T\). Taking the infimum over feasible solutions of \eqref{sdp:SigmaMinusPrimal} proves the inequality.
\end{IEEEproof}

\subsection{Exact Tensor Multiplicativity of \texorpdfstring{$\Sigma^{-}$}{Sigma-minus}}
Wang and Duan proved that \(\Sigma^{-}\) is supermultiplicative \cite{WangDuan2016}. The reverse inequality also holds:
\begin{theorem}\label{thm:SigmaMinusMult}
	For any Kraus operator spaces \(K_1,K_2\),
	\[
	\Sigma^{-}(K_1\otimes K_2)=\Sigma^{-}(K_1)\Sigma^{-}(K_2).
	\]
\end{theorem}

\begin{IEEEproof}
	For \(i=1,2\), let \((V_i,T_i)\) be primal feasible for \(\Sigma^{-}(K_i)\) in \eqref{sdp:SigmaMinusPrimal}, and define
	\(V_{12}:=\Pi_{12}(V_1\otimes V_2)\Pi_{12}^{\dagger}\) and \(T_{12}:=T_1\otimes T_2\).
	Then \(V_{12}\succeq0\), \(V_{12}=P_{12}V_{12}P_{12}\), and
	\(\operatorname{Tr}_{B_1B_2}V_{12} =(\operatorname{Tr}_{B_1}V_1)\otimes(\operatorname{Tr}_{B_2}V_2) \succeq I_{A_1A_2}\).
	Moreover, \(V_{12}\preceq I_{A_1A_2}\otimes T_{12}\). Thus \((V_{12},T_{12})\) is primal feasible for \(\Sigma^{-}(K_1\otimes K_2)\) in \eqref{sdp:SigmaMinusPrimal}, with objective value \(\operatorname{Tr} T_{12}=\operatorname{Tr} T_1\operatorname{Tr} T_2\). Taking the infimum over primal-feasible \((V_i,T_i)\) in \eqref{sdp:SigmaMinusPrimal} gives
	\[
	\Sigma^{-}(K_1\otimes K_2)\leq\Sigma^{-}(K_1)\Sigma^{-}(K_2).
	\]
	
	For the reverse inequality, let \((R_i,X_i)\) be feasible for \eqref{sdp:SigmaMinusDual}.
	
	Define \(R_{12}:=R_1\otimes R_2\) and \(X_{12}:=\Pi_{12}(X_1\otimes X_2)\Pi_{12}^{\dagger}\).
	The pair \((R_{12},X_{12})\) immediately satisfies the positivity and partial-trace constraints in the dual SDP \eqref{sdp:SigmaMinusDual} for \(\Sigma^{-}(K_1\otimes K_2)\). Let \(D_i:=X_i-R_i\otimes I_{B_i}\). Before reordering the tensor factors, we have
	\begin{align*}
		&(P_1\otimes P_2)
		\left[X_1\otimes X_2-(R_1\otimes I_{B_1})\otimes(R_2\otimes I_{B_2})\right]
		(P_1\otimes P_2)\\
		&=(P_1D_1P_1)\otimes(P_2X_2P_2)+
		\bigl(P_1(R_1\otimes I_{B_1})P_1\bigr)\otimes(P_2D_2P_2)
		\succeq0.
	\end{align*}
	Conjugating by \(\Pi_{12}\) gives
	\begin{align*}
		&\Pi_{12}(P_1\otimes P_2)
		\left[X_1\otimes X_2-(R_1\otimes I_{B_1})\otimes(R_2\otimes I_{B_2})\right]
		(P_1\otimes P_2)\Pi_{12}^\dagger
		\\&=P_{12}
		\left[X_{12}-(R_{12}\otimes I_{B_1B_2})\right]
		P_{12}\succeq0.
	\end{align*}
	Thus \((R_{12},X_{12})\) also satisfies the last constraint in \eqref{sdp:SigmaMinusDual}. It is therefore dual feasible, with objective value \(\operatorname{Tr} R_{12}=(\operatorname{Tr} R_1)(\operatorname{Tr} R_2)\). Taking the supremum over feasible \((R_i,X_i)\) for the dual SDPs \eqref{sdp:SigmaMinusDual} of \(\Sigma^{-}(K_i)\), and using strong duality, yields the reverse inequality.
\end{IEEEproof}

\subsection{The Bound \texorpdfstring{$\Sigma^{-}\leq\Gamma^2$}{Sigma-minus <= Gamma squared}}

\begin{lemma}\label{lem:WeightedPartialTraceCS}
	Let \(T_B\succeq0\), and let \(S_B\) be the orthogonal projection onto \(\operatorname{ran}T_B\). If
	\(X_{AB}=X_{AB}(I_A\otimes S_B)\), then
	\begin{equation}
		(\operatorname{Tr}_BX)(\operatorname{Tr}_BX)^{\dagger}
		\preceq
		(\operatorname{Tr}T_B)\,
		\operatorname{Tr}_B\!\left[X(I_A\otimes T_B^+)X^{\dagger}\right],
		\label{eq:WeightedPartialTraceCS}
	\end{equation}
	where \(T_B^+\) denotes the Moore--Penrose pseudoinverse.
\end{lemma}

\begin{IEEEproof}
	Write \(T_B=\sum_j t_j|j\rangle\langle j|\) in an eigenbasis, where \(t_j\geq0\), and write
	\(X=\sum_{i,j}X_{ij}\otimes|i\rangle\langle j|\), with \(X_{ij}\in\mathcal L(A)\). The assumption \(X=X(I_A\otimes S_B)\) implies \(X_{ij}=0\) whenever \(t_j=0\), and in particular \(X_{jj}=0\) for such \(j\).
	
	We have \(\operatorname{Tr}_BX=\sum_jX_{jj}\), while
	\begin{equation}
		\operatorname{Tr}_B\!\left[X(I_A\otimes T_B^+)X^{\dagger}\right]
		=
		\sum_{i,j:t_j>0}t_j^{-1}X_{ij}X_{ij}^{\dagger}
		\succeq
		\sum_{j:t_j>0}t_j^{-1}X_{jj}X_{jj}^{\dagger}.
		\label{eq:WeightedPartialTraceStep}
	\end{equation}
	For every \(|\xi\rangle\in A\), the ordinary Cauchy--Schwarz inequality gives
	\begin{align*}
		\left\|\left(\sum_jX_{jj}\right)^{\dagger}|\xi\rangle\right\|^2
		=
		\left\|\sum_{j:t_j>0}t_j^{1/2}\,
		\bigl(t_j^{-1/2}X_{jj}^{\dagger}|\xi\rangle\bigr)\right\|^2
		\leq
		\left(\sum_jt_j\right)
		\sum_{j:t_j>0}t_j^{-1}\|X_{jj}^{\dagger}|\xi\rangle\|^2.
	\end{align*}
	Equivalently,
	\(\left(\sum_jX_{jj}\right) \left(\sum_jX_{jj}\right)^{\dagger} \preceq (\operatorname{Tr}T_B) \sum_{j:t_j>0}t_j^{-1}X_{jj}X_{jj}^{\dagger}\).
	Combining this with \eqref{eq:WeightedPartialTraceStep} proves \eqref{eq:WeightedPartialTraceCS}.
\end{IEEEproof}

\begin{theorem}\label{thm:SigmaMinusGammaSquare}
	For every Kraus operator space \(K\),
	\[
	\Sigma^{-}(K)\leq\Gamma(K)^2.
	\]
\end{theorem}

\begin{IEEEproof}
	Let \((X,T_0,T_1)\) be any balanced feasible solution of \eqref{sdp:GammaEq} for \(\Gamma(K)\), and write $t:=\operatorname{Tr}T_0=\operatorname{Tr}T_1.$
	Let \(S_1\) denote the support projection of \(T_1\). Applying the generalized Schur complement to the block positive semidefinite constraint gives
	\(\operatorname{ran}X^{\dagger} \subseteq A\otimes\operatorname{ran}T_1,\; X(I_A\otimes T_1^+)X^{\dagger} \preceq I_A\otimes T_0\).
	The first condition is precisely the range condition required in the singular Schur-complement case and is equivalent to
	\(X=X(I_A\otimes S_1)\). Define
	\(F:=X(I_A\otimes T_1^+)X^{\dagger}\).
	Then \(F\succeq0\) and \(F\preceq I_A\otimes T_0\). Since \(PX=X\), we also have \(PF=FP=F\), and hence \(F=PFP\).
	
	The normalization condition gives \(\operatorname{Tr}_BX=I_A\). Applying Lemma~\ref{lem:WeightedPartialTraceCS} with \(T_1\) yields
	\[
	I_A
	=(\operatorname{Tr}_BX)(\operatorname{Tr}_BX)^{\dagger}
	\preceq
	t\,\operatorname{Tr}_BF.
	\]
	Set \(V:=tF\) and \(T':=tT_0\). Then \((V,T')\) is feasible for \eqref{sdp:SigmaMinusPrimal}, and
	\(\operatorname{Tr}T' =(\operatorname{Tr}T_1)(\operatorname{Tr}T_0) =t^2\).
	Thus every such feasible solution produces a feasible solution for \(\Sigma^{-}\) whose objective value is the square of the \(\Gamma\) objective. Taking infima and using Lemma~\ref{lem:GammaBalance} proves the claim.
\end{IEEEproof}

\subsection{General Bounds}

\begin{theorem}\label{thm:MainHierarchy}
	For every Kraus operator space \(K\),
	\[
	\widehat\Upsilon(K)
	\leq
	\Theta_{\NS}(K)
	\leq
	\Gamma(K)
	\leq
	\Sigma^{-}(K)
	\leq
	\Gamma(K)^2.
	\]
	Moreover, \(\Sigma^{-}(K)\leq\Sigma(K)\), and both \(\Gamma\) and \(\Sigma^{-}\) are exactly multiplicative under tensor products.
\end{theorem}

\begin{IEEEproof}
	The inequality \(\widehat\Upsilon(K)\leq\Theta_{\NS}(K)\) follows from \eqref{eq:ThetaWidehatReg}. The next three inequalities are Corollary~\ref{cor:ThetaGamma}, Proposition~\ref{prop:GammaSigmaMinus}, and Theorem~\ref{thm:SigmaMinusGammaSquare}, respectively. The comparison \(\Sigma^{-}\leq\Sigma\) follows from the inclusion of the feasible sets of \eqref{sdp:Sigma} and \eqref{sdp:SigmaMinusPrimal}. Exact multiplicativity follows from Theorems~\ref{thm:GammaMult} and~\ref{thm:SigmaMinusMult}.
\end{IEEEproof}

\subsection{Separation Between \texorpdfstring{$\log_2\Sigma^{-}$}{log2 Sigma-minus} and the Asymptotic Simulation Cost}
\label{subsec:SimulationSeparation}

Wang and Duan asked whether the lower bound
\(\log_2\Sigma^{-}(K)\leq S_{0,\NS}(K)\) is always an equality
\cite[Sec.~III]{WangDuan2016}. The following qubit family gives
a negative answer.

\begin{proposition}\label{prop:QubitSimulationSeparation}
	Let \(0<\eta\leq1/4\), and set
	\[
	E_0=
	\begin{pmatrix}
		\sqrt{1-\eta}&0\\
		0&\sqrt{1-\eta-\eta^2}
	\end{pmatrix},
	\quad
	E_1=
	\begin{pmatrix}
		0&\sqrt{\eta+\eta^2}\\
		\sqrt{\eta}&0
	\end{pmatrix}.
	\]
	Then \(K_\eta:=\spann\{E_0,E_1\}\) is a Kraus operator
	space from a qubit to a qubit, and
	\begin{equation}
		S_{0,\NS}(K_\eta)
		=2\log_2\!\left(\sqrt{1-\eta}
		+\sqrt{1-\eta-\eta^2}\right)
		>\log_2\Sigma^{-}(K_\eta).
		\label{eq:EpsilonSimulationGap}
	\end{equation}
	Moreover, as \(\eta\downarrow0\),
	\begin{equation}
		\widehat\Upsilon(K_\eta),\
		\Theta_{\NS}(K_\eta),\
		\Gamma(K_\eta),\
		\Sigma^{-}(K_\eta)\to2,
		\quad
		\Sigma(K_\eta)\to4.
		\label{eq:EpsilonParameterLimits}
	\end{equation}
	Consequently,
	\(S_{0,\NS}(K_\eta)
	-\log_2\Sigma^{-}(K_\eta)\to1\).
\end{proposition}

\begin{IEEEproof}
	The identity \(E_0^\dagger E_0+E_1^\dagger E_1=I_2\) shows
	that these operators define a channel.
	
	Write
	\(q:=\sqrt{1-\eta}\),
	\(h:=\sqrt{1-\eta-\eta^2}\),
	and let \(v_0:=q|00\rangle+h|11\rangle\),
	\(v_1:=\sqrt{\eta}|01\rangle
	+\sqrt{\eta+\eta^2}|10\rangle\).
	
	Its Choi matrix is
	\(J:=|v_0\rangle\langle v_0|+|v_1\rangle\langle v_1|\),
	and its support projection is
	\(P:=\frac{|v_0\rangle\langle v_0|}{q^2+h^2}
	+\frac{|v_1\rangle\langle v_1|}{2-q^2-h^2}\).
	
	The four operators \(E_i^\dagger E_j\), \(i,j\in\{0,1\}\),
	are linearly independent. We first use this independence to determine the Choi variable
	in \eqref{sdp:Sigma}. Any \(F\succeq0\) satisfying \(F=PFP\)
	can be written as
	\(F=\sum_{i,j=0}^1 c_{ij}|v_i\rangle\langle v_j|\)
	with \((c_{ij})\succeq0\).
	The constraint \(\operatorname{Tr}_B F=I_A\) is equivalent to
	\(\sum_{i,j=0}^1 c_{ij}E_j^\dagger E_i
	=I_A=\sum_{i=0}^1 E_i^\dagger E_i\).
	Linear independence therefore forces \(c_{ij}=\delta_{ij}\),
	so \(F=J\).
	
	For \(n\geq1\) and \(s\in\{0,1\}^n\), let
	\(E_s:=E_{s_1}\otimes\cdots\otimes E_{s_n}\).
	The operators
	\(E_s^\dagger E_t
	=\bigotimes_{\ell=1}^n(E_{s_\ell}^\dagger E_{t_\ell})\)
	are again linearly independent.
	Extending the notation \(\Pi_{12}\), let
	\(\Pi_n:(A\otimes B)^{\otimes n}
	\to A^{\otimes n}\otimes B^{\otimes n}\)
	denote the reordering unitary that places all \(A\) factors
	before all \(B\) factors, preserving their order within each group.
	The Choi matrix and support projection of the tensor-power channel
	are \(J_n:=\Pi_n J^{\otimes n}\Pi_n^\dagger\) and
	\(\Pi_n P^{\otimes n}\Pi_n^\dagger\), respectively.
	Applying the same trace-preservation argument shows that
	the Choi variable in \eqref{sdp:Sigma} for
	\(K_\eta^{\otimes n}\) must equal \(J_n\).
	
	Thus this program reduces to
	\[
	\Sigma(K_\eta^{\otimes n})
	=
	\min\left\{
	\operatorname{Tr}Y:
	Y\succeq0,\ 
	J_n\preceq I_{A^{\otimes n}}\otimes Y
	\right\},
	\]
	where \(Y\) acts on \(B^{\otimes n}\).
	
	To obtain an upper bound, take
	\(T:=\operatorname{diag}(q^2+qh,h^2+qh)\).
	In the ordered basis
	\((|00\rangle,|01\rangle,|10\rangle,|11\rangle)\),
	\[
	I_A\otimes T-J
	=
	\begin{pmatrix}
		qh&0&0&-qh\\
		0&q^2+h^2-1+qh&-\sqrt{\eta(\eta+\eta^2)}&0\\
		0&-\sqrt{\eta(\eta+\eta^2)}&q^2+h^2-1+qh&0\\
		-qh&0&0&qh
	\end{pmatrix}.
	\]
	This matrix is positive semidefinite, since
	\(q^2+h^2>1\) and
	\(q^2h^2-\eta(\eta+\eta^2)=q^2+h^2-1>0\).
	Thus \(0\preceq J\preceq I_A\otimes T\).
	Taking tensor powers and conjugating by \(\Pi_n\) gives
	\(J_n\preceq I_{A^{\otimes n}}\otimes T^{\otimes n}\).
	Hence \(Y=T^{\otimes n}\) is feasible in the reduced program,
	and \(\Sigma(K_\eta^{\otimes n})
	\leq\operatorname{Tr}(T^{\otimes n})
	=(\operatorname{Tr}T)^n=(q+h)^{2n}\).
	
	For the matching lower bound, let
	\(|\Phi\rangle:=|00\rangle+|11\rangle\) and
	\(|\Phi_n\rangle:=\Pi_n|\Phi\rangle^{\otimes n}
	=\sum_{s\in\{0,1\}^n}
	|s\rangle_{A^{\otimes n}}\otimes|s\rangle_{B^{\otimes n}}\).
	Since \(\langle\Phi|v_0\rangle=q+h\) and
	\(\langle\Phi|v_1\rangle=0\), every feasible \(Y\) satisfies
	\[
	\begin{aligned}
		\operatorname{Tr}Y=\langle\Phi_n|
		I_{A^{\otimes n}}\otimes Y
		|\Phi_n\rangle
		\geq\langle\Phi_n|J_n|\Phi_n\rangle
		=\bigl(\langle\Phi|J|\Phi\rangle\bigr)^n
		=(q+h)^{2n}.
	\end{aligned}
	\]
	The two bounds give
	\(\Sigma(K_\eta^{\otimes n})=(q+h)^{2n}\)
	for every \(n\geq1\).
	By the definition of the asymptotic simulation cost,
	\(S_{0,\NS}(K_\eta)
	=\lim_{n\to\infty}\frac1n
	\log_2\Sigma(K_\eta^{\otimes n})
	=2\log_2(q+h)\).
	
	Direct calculation gives
	\(m_\eta:=\lambda_{\min}(\operatorname{Tr}_B P)
	=\frac{1-\eta}{2-2\eta-\eta^2}+\frac{1}{2+\eta}\)
	and
	\(M_\eta:=\|\operatorname{Tr}_A P\|_\infty
	=\frac{1-\eta}{2-2\eta-\eta^2}+\frac{1+\eta}{2+\eta}\).
	The choice \(V=P/m_\eta\), \(T_B=I_2/m_\eta\)
	is feasible for \eqref{sdp:SigmaMinusPrimal}.
	The choice \(S_A=I_2/M_\eta\),
	\(U_{AB}=P/M_\eta\) is feasible for the primal SDP \eqref{sdp:WidehatUpsilonMax}
	for \(\widehat\Upsilon\). Theorem~\ref{thm:MainHierarchy}
	then yields
	\[
	\frac{2}{M_\eta}
	\leq\widehat\Upsilon(K_\eta)
	\leq\Theta_{\NS}(K_\eta)
	\leq\Gamma(K_\eta)
	\leq\Sigma^{-}(K_\eta)
	\leq\frac{2}{m_\eta}.
	\]
	For \(0<\eta\leq1/4\), we have
	\(m_\eta>1/2+4/9=17/18\) and
	\((q+h)^2>4h^2\geq11/4\).
	Thus \(\Sigma^{-}(K_\eta)<\frac{36}{17}
	<\frac{11}{4}\leq\Sigma(K_\eta)\),
	which proves \eqref{eq:EpsilonSimulationGap}.
	Finally, \(m_\eta,M_\eta\to1\) and
	\((q+h)^2\to4\), proving the limits.
\end{IEEEproof}

For this family, \(\log_2\Gamma(K_\eta)
-C_{0,\NS}(K_\eta)\to0\), whereas
\(S_{0,\NS}(K_\eta)-\log_2\Sigma^{-}(K_\eta)\to1\).
Thus the capacity upper bounds approach the capacity as
\(\eta\downarrow0\), while the simulation lower bound
retains a nonzero gap.

\section{Sharp Bounds for Classical--Quantum Kraus Operator Spaces}
\label{sec:CqSharp}

Throughout this section, \(K\) is a classical--quantum Kraus operator space with output support projections \(\{P_x\}_{x\in\mathcal X}\). We write \(r:=\dim B'\) for the effective output dimension.

\subsection{Classical--Quantum Reductions}

\begin{proposition}\label{prop:CqReductions}
	For every classical--quantum Kraus operator space \(K\),
	\begin{equation}
		\begin{aligned}
			\Sigma(K)=\Sigma^{-}(K)=\min\quad
			&\operatorname{Tr}T_B\\
			\mathrm{s.t.}\quad
			&T_B\succeq0,\\
			&0\preceq F_x\preceq P_x,\quad \operatorname{Tr}F_x=1,\\
			&F_x\preceq T_B\quad \forall x.
		\end{aligned}
		\label{sdp:CqSigma}
	\end{equation}
	Moreover, the primal and dual programs for \(\Gamma\) reduce to
	\begin{equation}
		\begin{aligned}
			\Gamma(K)=\min\quad
			&\frac12\operatorname{Tr}(T_0+T_1)\\
			\mathrm{s.t.}\quad
			&P_xX_x=X_x,\quad \operatorname{Tr}X_x=1,\\
			&\begin{pmatrix}T_0&X_x\\X_x^{\dagger}&T_1\end{pmatrix}\succeq0
			\quad \forall x,
		\end{aligned}
		\label{sdp:CqGammaPrimal}
	\end{equation}
	and
	\begin{equation}
		\begin{aligned}
			\Gamma(K)=\max\quad
			&\sum_x g_x\\
			\mathrm{s.t.}\quad
			&g_x\geq0,\\
			&\begin{pmatrix}Y_{0,x}&Z_x\\Z_x^{\dagger}&Y_{1,x}\end{pmatrix}\succeq0
			\quad \forall x,\\
			&\sum_xY_{0,x}=\sum_xY_{1,x}=I_B,\\
			&P_xZ_x=g_xP_x\quad \forall x.
		\end{aligned}
		\label{sdp:CqGammaDual}
	\end{equation}
	Finally, the dual of \(\Sigma\) is
	\begin{equation}
		\begin{aligned}
			\Sigma(K)=\max\quad
			&\sum_x\lambda_x\\
			\mathrm{s.t.}\quad
			&Y_x\succeq0,\quad \sum_xY_x\preceq I_B,\\
			&P_xY_xP_x\succeq\lambda_xP_x,\quad \lambda_x\geq0.
		\end{aligned}
		\label{sdp:CqSigmaDual}
	\end{equation}
	All displayed programs satisfy strong duality.
\end{proposition}

\begin{IEEEproof}
	The reductions follow by pinching the classical input register. The details, including the dual programs and strong duality, are given in Appendix~\ref{app:CqReductions}.
\end{IEEEproof}

\subsection{Projection Inequalities}

\begin{lemma}\label{lem:ProjectionResource}
	Let \(B\) be an \(r\)-dimensional Hilbert space.
	
	\emph{(i)} If \(Y_x\succeq0\), \(\sum_xY_x\preceq I_B\), and \(P_x\) are arbitrary projections, then
	\begin{equation}
		\sum_xP_xY_xP_x
		\preceq
		\left(\frac{1+\sqrt r}{2}\right)^2I_B.
		\label{eq:RankFreeCompression}
	\end{equation}
	
	\emph{(ii)} Let \(P\) be a projection, let \(Y\succeq0\), and suppose that
	\(PYP\succeq\lambda P\) and \(\operatorname{rank}P\geq k\). For any unit vector \(v\), set
	\(p:=\langle v|P|v\rangle\), \(a:=\langle v|Y|v\rangle\), and
	\(b:=\operatorname{Tr}\!\left[(I-|v\rangle\langle v|)Y\right]\).
	Define
	\(\alpha_{r,k}:=\frac{(1+\sqrt r)^2}{\sqrt r(1+\sqrt k)^2},\; \beta_{r,k}:=\frac{1+\sqrt r}{\sqrt r(1+\sqrt k)^2}\).
	Then
	\begin{equation}
		\lambda p\leq\alpha_{r,k}a+\beta_{r,k}b.
		\label{eq:SupportFloor}
	\end{equation}
\end{lemma}

\begin{IEEEproof}
	We first prove part (i). Fix a unit vector \(v\). For each \(x\), set
	\(p_x:=\langle v|P_x|v\rangle\), \(a_x:=\langle v|Y_x|v\rangle\), and
	\(b_x:=\operatorname{Tr}\!\left[(I-|v\rangle\langle v|)Y_x\right]\).
	If \(0<p_x<1\), write
	\(P_xv=p_xv+\sqrt{p_x(1-p_x)}\,w_x,\; w_x\perp v,\; \|w_x\|=1\).
	Let \(c_x:=\langle w_x|Y_x|w_x\rangle\). Since \(Y_x\succeq0\) and \(w_x\in v^\perp\), we have
	\(c_x\leq b_x\) and \(|\langle v|Y_x|w_x\rangle|\leq\sqrt{a_xc_x}\), and hence
	\begin{align}
		\langle v|P_xY_xP_x|v\rangle
		&\leq
		\left(p_x\sqrt{a_x}+\sqrt{p_x(1-p_x)}\sqrt{c_x}\right)^2\nonumber\\
		&\leq
		\left(p_x\sqrt{a_x}+\sqrt{p_x(1-p_x)}\sqrt{b_x}\right)^2.
		\label{eq:CompressionLocal}
	\end{align}
	The same estimate is immediate for \(p_x=0\) and \(p_x=1\).
	
	Set
	\(
	\alpha_r:=\frac{(1+\sqrt r)^2}{4\sqrt r},
	\;
	\beta_r:=\frac{1+\sqrt r}{4\sqrt r}.
	\)
	For every \(p\in[0,1]\),
	\[
	\frac{p^2}{\alpha_r}+\frac{p(1-p)}{\beta_r}
	=
	\frac{4\sqrt r\,p(1+\sqrt r-p\sqrt r)}{(1+\sqrt r)^2}
	\leq1,
	\]
	because the quadratic expression on the right attains its maximum \(1\) at \(p=(1+\sqrt r)/(2\sqrt r)\). Weighted Cauchy--Schwarz therefore gives
	\[
	\left(p\sqrt a+\sqrt{p(1-p)}\sqrt b\right)^2
	\leq
	\left(\frac{p^2}{\alpha_r}+\frac{p(1-p)}{\beta_r}\right)
	(\alpha_ra+\beta_rb)
	\leq\alpha_ra+\beta_rb.
	\]
	The constraint \(\sum_xY_x\preceq I_B\) implies
	\(\sum_xa_x\leq1,\; \sum_xb_x\leq\operatorname{Tr}(I-|v\rangle\langle v|)=r-1\).
	Combining these estimates with \eqref{eq:CompressionLocal} and the weighted Cauchy--Schwarz bound above gives
	\[
	\left\langle v\left|\sum_xP_xY_xP_x\right|v\right\rangle
	\leq
	\alpha_r+(r-1)\beta_r
	=
	\frac{(1+\sqrt r)^2}{4}.
	\]
	Since \(v\) is arbitrary, \eqref{eq:RankFreeCompression} follows.
	
	We now prove part (ii). If \(\lambda\leq0\), then \(\lambda p\leq0\leq\alpha_{r,k}a+\beta_{r,k}b\), so the claim is immediate. Hence assume \(\lambda>0\). If \(p=0\), the left-hand side of \eqref{eq:SupportFloor} vanishes. If \(p=1\), then \(v\in\operatorname{ran}P\), so \(a\geq\lambda\). Moreover, \(\operatorname{ran}P\cap v^{\perp}\) contains \(k-1\) orthonormal vectors, and the hypothesis \(PYP\succeq\lambda P\) gives \(b\geq(k-1)\lambda\). Therefore,
	\[
	\alpha_{r,k}a+\beta_{r,k}b
	\geq
	[\alpha_{r,k}+(k-1)\beta_{r,k}]\lambda
	\geq\lambda,
	\]
	where the last inequality is equivalent to \((\sqrt r-\sqrt k)^2\geq0\).
	
	It remains to consider \(0<p<1\). Define
	\(u:=Pv/\sqrt p=\sqrt p\,v+\sqrt{1-p}\,w\), where \(w\perp v\).
	Then \(u\) is a unit vector in \(\operatorname{ran}P\). Since \(\operatorname{rank}P\geq k\), the subspace of \(\operatorname{ran}P\) orthogonal to \(u\) has dimension at least \(k-1\). Choose orthonormal vectors \(e_1,\ldots,e_{k-1}\) in this subspace.
	
	These vectors are automatically orthogonal to \(v\). Indeed, since \(e_j\in\operatorname{ran}P\),
	\(
	\langle e_j|v\rangle
	=
	\langle e_j|Pv\rangle
	=
	\sqrt p\,\langle e_j|u\rangle
	=0.
	\)
	The decomposition of \(u\), together with \(e_j\perp u\) and \(e_j\perp v\), then implies \(e_j\perp w\). Thus
	\(w,e_1,\ldots,e_{k-1}\) are mutually orthogonal and all lie in \(v^{\perp}\).
	
	From \(PYP\succeq\lambda P\),
	\(
	\langle u|Y|u\rangle\geq\lambda,
	\;
	\langle e_j|Y|e_j\rangle\geq\lambda.
	\)
	Let \(c:=\langle w|Y|w\rangle\). Since \(w,e_1,\ldots,e_{k-1}\) are mutually orthogonal and lie in \(v^{\perp}\),
	\(b\geq c+(k-1)\lambda\).
	Positivity of \(Y\) also gives \(|\langle v|Y|w\rangle|\leq\sqrt{ac}\), and hence
	\begin{equation}
		\lambda
		\leq
		\langle u|Y|u\rangle
		\leq
		\left(\sqrt{pa}+\sqrt{(1-p)c}\right)^2.
		\label{eq:LambdaTwoDim}
	\end{equation}
	
	For fixed \(p\) and \(\lambda\), weighted Cauchy--Schwarz applied to \eqref{eq:LambdaTwoDim} implies
	\(\alpha_{r,k}a+\beta_{r,k}c \geq \frac{\lambda}{p/\alpha_{r,k}+(1-p)/\beta_{r,k}}\).
	Combining this with \(b\geq c+(k-1)\lambda\) gives
	\begin{align*}
		\alpha_{r,k}a+\beta_{r,k}b-\lambda p
		\geq
		\lambda\left[
		\frac{1}{p/\alpha_{r,k}+(1-p)/\beta_{r,k}}
		+\beta_{r,k}(k-1)-p
		\right].
	\end{align*}
	Substituting the definitions of \(\alpha_{r,k}\) and \(\beta_{r,k}\), the bracket equals
	\begin{equation}
		\frac{[p\sqrt r(1+\sqrt k)-\sqrt k(1+\sqrt r)]^2}
		{\sqrt r(1+\sqrt k)^2(1+\sqrt r-p\sqrt r)}.
		\label{eq:LocalSlack}
	\end{equation}
	Because \(0\leq p\leq1\) and \(r\geq1\),
	\(1+\sqrt r-p\sqrt r=1+\sqrt r(1-p)>0\). Thus \eqref{eq:LocalSlack} is nonnegative, proving \eqref{eq:SupportFloor}. The expression in \eqref{eq:LocalSlack} vanishes at
	\(
	p=\frac{\sqrt k(1+\sqrt r)}{\sqrt r(1+\sqrt k)}.
	\)
	For \(k<r\), this point lies in \((0,1)\); for \(k=r\), it is the endpoint \(p=1\).
\end{IEEEproof}

\subsection{Dimension-Dependent Bounds Between \texorpdfstring{$A$}{A} and \texorpdfstring{$\Gamma$}{Gamma}}

Set \(c_r:=(1+\sqrt r)/2\).

\begin{theorem}\label{thm:GammaOverA}
	For every classical--quantum Kraus operator space with effective output dimension \(r\),
	\[
	A(K)\leq\Gamma(K)\leq c_rA(K).
	\]
\end{theorem}

\begin{IEEEproof}
	The lower bound follows from \(A(K)=\Theta_{\NS}(K)\) and Corollary~\ref{cor:ThetaGamma}. We prove the upper bound.
	
	Let \((g_x,Y_{0,x},Y_{1,x},Z_x)_x\) be any feasible solution of \eqref{sdp:CqGammaDual}. Fix a unit vector \(v\), and set \(p_x:=\langle v|P_x|v\rangle\). Compressing the positive semidefinite block matrix to the pair of vectors \((P_xv,v)\) gives
	\[
	|\langle P_xv|Z_x|v\rangle|^2
	\leq
	\langle v|P_xY_{0,x}P_x|v\rangle
	\langle v|Y_{1,x}|v\rangle.
	\]
	Since \(P_xZ_x=g_xP_x\) and \(g_x\geq0\), the square root of the left-hand side is \(g_xp_x\). Summing over \(x\) and applying scalar Cauchy--Schwarz yields
	\begin{align*}
		\left\langle v\left|\sum_xg_xP_x\right|v\right\rangle
		&=\sum_xg_xp_x\leq
		\left[\sum_x\langle v|P_xY_{0,x}P_x|v\rangle\right]^{1/2}
		\left[\sum_x\langle v|Y_{1,x}|v\rangle\right]^{1/2}.
	\end{align*}
	The second factor equals \(1\). Applying Lemma~\ref{lem:ProjectionResource}(i) with \(Y_x=Y_{0,x}\), using \(\sum_xY_{0,x}=I_B\), gives
	\(\sum_xP_xY_{0,x}P_x\preceq c_r^2I_B\).
	Hence \(\langle v|\sum_xg_xP_x|v\rangle\leq c_r\). Since \(v\) is arbitrary,
	\(\sum_x(g_x/c_r)P_x\preceq I_B\). Therefore \(s_x:=g_x/c_r\) is feasible for the packing SDP \eqref{sdp:SemidefinitePackingPrimal}, and
	\(A(K)\geq\frac1{c_r}\sum_xg_x\).
	Taking the supremum over dual-feasible solutions of \(\Gamma\) proves \(\Gamma(K)\leq c_rA(K)\).
\end{IEEEproof}

\subsection{Rank-Sensitive Bounds on \texorpdfstring{$\Sigma/A$}{Sigma/A} and \texorpdfstring{$\Sigma/\Gamma$}{Sigma/Gamma}}

\begin{theorem}\label{thm:RankSensitiveSigmaA}
    Let \(k:=\min_x\operatorname{rank}P_x,\;c_{r,k}:=\frac{1+\sqrt r}{1+\sqrt k}.\)
	For every such classical--quantum Kraus operator space,
	\begin{equation}
		\frac{\Sigma(K)}{A(K)}
		\leq
		c_{r,k}^2
		=
		\left(\frac{1+\sqrt r}{1+\sqrt k}\right)^2.
		\label{eq:RankSensitiveSigmaA}
	\end{equation}
	For every pair \(1\leq k\leq r\), there exists a classical--quantum Kraus operator space attaining equality.
\end{theorem}

\begin{IEEEproof}
	Let \((\lambda_x,Y_x)_x\) be any feasible solution of the dual SDP \eqref{sdp:CqSigmaDual}. Fix a unit vector \(v\), and for every \(x\) define
	\(
	p_x:=\langle v|P_x|v\rangle,
	\;
	a_x:=\langle v|Y_x|v\rangle,
	\;
	b_x:=\operatorname{Tr}\!\left[(I-|v\rangle\langle v|)Y_x\right].
	\)
	Lemma~\ref{lem:ProjectionResource}(ii) gives
	\(\lambda_xp_x\leq\alpha_{r,k}a_x+\beta_{r,k}b_x\). The constraint \(\sum_xY_x\preceq I_B\) implies
	\(
	\sum_xa_x\leq1,
	\;
	\sum_xb_x\leq r-1.
	\)
	Therefore, 
	\begin{align*}
		\left\langle v\left|\sum_x\lambda_xP_x\right|v\right\rangle
		&\leq\alpha_{r,k}+(r-1)\beta_{r,k}=\left(\frac{1+\sqrt r}{1+\sqrt k}\right)^2
		=c_{r,k}^2.
	\end{align*}
	Since \(v\) is arbitrary,
	\[\sum_x\lambda_xP_x\preceq c_{r,k}^2I_B.\] 
    Thus \(s_x:=\lambda_x/c_{r,k}^2\) is feasible for the packing primal, so
	\(
	A(K)\geq\frac1{c_{r,k}^2}\sum_x\lambda_x.
	\)
	Taking the supremum over \eqref{sdp:CqSigmaDual} proves \eqref{eq:RankSensitiveSigmaA}. The sharpness construction is given in Proposition~\ref{prop:RankKFamily}.
\end{IEEEproof}

\begin{theorem}\label{thm:RankSensitiveSigmaGamma}
	Under the same assumptions,
	\begin{equation}
		\frac{\Sigma(K)}{\Gamma(K)}
		\leq
		c_{r,k}
		=
		\frac{1+\sqrt r}{1+\sqrt k}.
		\label{eq:RankSensitiveSigmaGamma}
	\end{equation}
\end{theorem}

\begin{IEEEproof}
	Let \((\lambda_x,Y_x)_x\) again be any dual-feasible solution for \(\Sigma\), and define
	\(
	C_x:=P_xY_xP_x\succeq\lambda_xP_x.
	\)
	The proof of Theorem~\ref{thm:RankSensitiveSigmaA} gives the operator inequality
	\begin{equation}
		\sum_x\lambda_xP_x\preceq c_{r,k}^2I_B.
		\label{eq:LambdaPCrk}
	\end{equation}
	If \(\lambda_x>0\), then \(C_x\) is strictly positive on \(\operatorname{ran}P_x\), and hence
	\(C_x^+\preceq\lambda_x^{-1}P_x\). Define
	\(
	s_x:=\frac{\lambda_x}{c_{r,k}},
	\;
	D_x:=\frac{\lambda_x^2}{c_{r,k}^2}C_x^+,
	\;
	Z_x:=\frac{\lambda_x}{c_{r,k}}Y_xP_xC_x^+.
	\)
	If \(\lambda_x=0\), set \(s_x=0\) and \(D_x=Z_x=0\).
	
	For \(\lambda_x>0\), \(C_xC_x^+=P_x\), and therefore
	\(P_xZ_x=s_xP_x\). Moreover,
	\(\operatorname{ran}Z_x\subseteq\operatorname{ran}Y_x\), and
	\begin{align*}
		Z_x^{\dagger}Y_x^+Z_x
		&=\frac{\lambda_x^2}{c_{r,k}^2}
		C_x^+P_xY_xY_x^+Y_xP_xC_x^+=\frac{\lambda_x^2}{c_{r,k}^2}C_x^+C_xC_x^+
		=D_x.
	\end{align*}
	Thus the range condition and the Schur-complement inequality are satisfied, so the generalized Schur complement gives
	\(
	\begin{psmallmatrix}
		Y_x&Z_x\\
		Z_x^{\dagger}&D_x
	\end{psmallmatrix}
	\succeq0.
	\)
	Also,
	\(D_x\preceq(\lambda_x/c_{r,k}^2)P_x\), so \eqref{eq:LambdaPCrk} gives \(\sum_xD_x\preceq I_B\); the original dual solution satisfies \(\sum_xY_x\preceq I_B\).
	
	Appendix~\ref{app:CqReductions} shows that the two normalization equalities in \eqref{sdp:CqGammaDual} may equivalently be relaxed to
	\(\sum_xY_{0,x}\preceq I_B\) and
	\(\sum_xY_{1,x}\preceq I_B\): any deficits can be added to diagonal blocks without changing \(Z_x\) or the objective. Hence
	\((s_x,Y_x,D_x,Z_x)_x\) gives a feasible solution of the dual of \(\Gamma\), with objective value
	\(
	\sum_xs_x=\frac1{c_{r,k}}\sum_x\lambda_x.
	\)
	Taking the supremum over dual-feasible solutions for \(\Sigma\) yields
	\(\Gamma(K)\geq\Sigma(K)/c_{r,k}\), which is \eqref{eq:RankSensitiveSigmaGamma}.
\end{IEEEproof}

Since \(k\geq1\), we have \(c_{r,k}\leq c_r\). Theorems~\ref{thm:RankSensitiveSigmaA} and \ref{thm:RankSensitiveSigmaGamma} therefore give
\(\Sigma(K)\leq c_r\Gamma(K),\; \frac{\Sigma(K)}{A(K)}\leq c_r^2\).
Combining these bounds with Theorem~\ref{thm:GammaOverA} and the general inequality \(\Gamma(K)\leq\Sigma(K)\) yields
\begin{equation}
	A(K)\leq\Gamma(K)\leq c_rA(K),
	\quad
	\Gamma(K)\leq\Sigma(K)\leq c_r\Gamma(K).
	\label{eq:CqTwoSegments}
\end{equation}
The following pure-state family attains equality in both dimension-dependent bounds in \eqref{eq:CqTwoSegments}.

\subsection{An \texorpdfstring{$r$}{r}-Dimensional Pure-State Extremal Family}

For \(r=1\), all parameters equal \(1\), so assume \(r\geq2\).

\begin{proposition}\label{prop:CanonicalPureFamily}
	Set \(p:=(1+\sqrt r)/(2\sqrt r)\).
	Choose an orthonormal basis \(\{v,e_1,\ldots,e_{r-1}\}\) of an \(r\)-dimensional Hilbert space and define
	\begin{equation}
		|\psi_{j,\varepsilon}\rangle
		:=\sqrt p\,|v\rangle
		+\varepsilon\sqrt{1-p}\,|e_j\rangle,
		\quad
		P_{j,\varepsilon}:=|\psi_{j,\varepsilon}\rangle\langle\psi_{j,\varepsilon}|,
		\label{eq:CanonicalFamily}
	\end{equation}
	where \(j=1,\ldots,r-1\) and \(\varepsilon\in\{+1,-1\}\). The corresponding classical--quantum Kraus operator space satisfies
	\[
	A=\frac{2\sqrt r}{1+\sqrt r},
	\quad
	\Gamma=\sqrt r,
	\quad
	\Sigma=\frac{\sqrt r(1+\sqrt r)}2.
	\]
	Consequently,
	\begin{equation}
		\frac{\Gamma}{A}
		=
		\frac{\Sigma}{\Gamma}
		=
		\frac{1+\sqrt r}{2},
		\quad
		\frac{\Sigma}{A}
		=
		\left(\frac{1+\sqrt r}{2}\right)^2.
		\label{eq:CanonicalRatios}
	\end{equation}
\end{proposition}

\begin{IEEEproof}
	We first compute \(A\). Any packing weights \(s_{j,\varepsilon}\) satisfy
	\(\sum_{j,\varepsilon}s_{j,\varepsilon}P_{j,\varepsilon}\preceq I\). Taking the expectation in \(v\) gives
	\(p\sum_{j,\varepsilon}s_{j,\varepsilon}\leq1\),
	so \(A\leq1/p\).
	
	Conversely, set \(s_{j,\varepsilon}=1/[2(r-1)p]\).
	Summing over the two signs cancels the cross terms, giving
	\[
	\sum_{j,\varepsilon}s_{j,\varepsilon}P_{j,\varepsilon}
	=
	|v\rangle\langle v|
	+
	\frac{1-p}{p(r-1)}(I-|v\rangle\langle v|).
	\]
	Since \((1-p)/[p(r-1)]=1/(1+\sqrt r)^2\leq1\),
	these weights are feasible for \eqref{sdp:SemidefinitePackingPrimal}
	and have total weight \(1/p\). Hence
	\[
	A=\frac1p=\frac{2\sqrt r}{1+\sqrt r}.
	\]
	
	We next compute \(\Sigma\). Since each \(P_{j,\varepsilon}\) has rank one, the conditions \(0\preceq F_{j,\varepsilon}\preceq P_{j,\varepsilon}\) and \(\operatorname{Tr}F_{j,\varepsilon}=1\) force \(F_{j,\varepsilon}=P_{j,\varepsilon}\). Thus the constraints in \eqref{sdp:CqSigma} reduce to \(P_{j,\varepsilon}\preceq T\) for all \(j,\varepsilon\).
	
	The family is invariant under sign changes and permutations of \(e_1,\ldots,e_{r-1}\), with \(v\) fixed. Averaging a feasible \(T\) over these transformations preserves feasibility and trace, eliminates the off-diagonal entries, and makes the diagonal entries on \(v^\perp\) equal. Hence we may take \(T=a|v\rangle\langle v|+b(I-|v\rangle\langle v|)\), with \(a,b>0\) and objective value \(a+(r-1)b\).
	For \(T\succ0\), the Schur complement shows that
	\(P_{j,\varepsilon}\preceq T\) is equivalent to
	\(\langle\psi_{j,\varepsilon}|T^{-1}|\psi_{j,\varepsilon}\rangle\leq1\).
	Using \eqref{eq:CanonicalFamily} and the above form of \(T\), this becomes
	\begin{equation}
		\frac pa+\frac{1-p}{b}\leq1.
		\label{eq:CanonicalCoverConstraint}
	\end{equation}
	
	Cauchy--Schwarz gives
	\begin{align*}
		[a+(r-1)b]
		\left[\frac pa+\frac{1-p}{b}\right]
		\geq
		\left[\sqrt p+\sqrt{(r-1)(1-p)}\right]^2.
	\end{align*}
	Since \(p=(1+\sqrt r)/(2\sqrt r)\),
	\((r-1)(1-p)=(\sqrt r-1)^2p\), and the right-hand side equals
	\(rp=\sqrt r(1+\sqrt r)/2\). Thus every feasible \(T\) satisfies
	\(\operatorname{Tr}T\geq\sqrt r(1+\sqrt r)/2\).
	
	Taking \(a=\frac{1+\sqrt r}{2}\) and \(b=\frac12\)
	satisfies \eqref{eq:CanonicalCoverConstraint} with equality
	and gives
	\(\operatorname{Tr}T=a+(r-1)b=\sqrt r(1+\sqrt r)/2\).
	This attains the lower bound, so
	\(\Sigma=\sqrt r(1+\sqrt r)/2\). Finally, Theorem~\ref{thm:GammaOverA} and \eqref{eq:CqTwoSegments} give
	\(\Gamma\leq c_rA=\sqrt r,\; \Gamma\geq\frac{\Sigma}{c_r}=\sqrt r\).
	Hence \(\Gamma=\sqrt r\), and \eqref{eq:CanonicalRatios} follows.
\end{IEEEproof}

\subsection{A Rank-\texorpdfstring{$k$}{k} Extremal Family}

\begin{proposition}\label{prop:RankKFamily}
	For every pair of integers \(1\leq k\leq r\), there exists a classical--quantum Kraus operator space with effective output dimension \(r\) and all output support projections of rank \(k\) such that
	\[
	\frac{\Sigma(K)}{A(K)}
	=
	\left(\frac{1+\sqrt r}{1+\sqrt k}\right)^2.
	\]
	Hence the constant in \eqref{eq:RankSensitiveSigmaA} is optimal for every \((r,k)\).
\end{proposition}

\begin{IEEEproof}
	Assume first that \(1\leq k<r\). Set
	\(p:=\sqrt k(1+\sqrt r)/[\sqrt r(1+\sqrt k)]\).
	Choose a unit vector \(v\in B\), set \(W:=v^\perp\), and let \(Q_v:=I-|v\rangle\langle v|\).
    Then \(B=\mathbb Cv\oplus W\) and \(\dim W=r-1\). Fix an orthonormal basis \(e_1,\ldots,e_{r-1}\) of \(W\), set \(w:=e_1\), and let
	\(E:=\spann\{e_2,\ldots,e_k\}\),
	with the convention \(E=\{0\}\) when \(k=1\). Let \(P_E\) denote the orthogonal projection onto \(E\), and define
	\(u:=\sqrt p\,v+\sqrt{1-p}\,w,\; P_0:=|u\rangle\langle u|+P_E\).
	Then \(\operatorname{rank}P_0=k\).
	
	Let \(\mathcal G\) be the finite signed-permutation group on \(W\) with respect to the basis \(\{e_j\}\): every \(U_g\), \(g\in\mathcal G\), permutes the basis vectors and independently multiplies them by signs \(\pm1\). Extend it to \(B\) by \(\widetilde U_g:=1\oplus U_g\), and take
	\(P_g:=\widetilde U_gP_0\widetilde U_g^{\dagger}\)
	as the output support projections. For each \(g\in\mathcal G\), set \(w_g:=U_gw\) and \(E_g:=U_gE\). Averaging over the independent sign changes eliminates the off-diagonal entries in the basis \(\{v,e_1,\ldots,e_{r-1}\}\), while averaging over permutations makes the diagonal entries on \(W\) equal. Thus
	\[
	\overline P
	:=\frac1{|\mathcal G|}\sum_{g\in\mathcal G}P_g
	=
	p|v\rangle\langle v|
	+
	\frac{k-p}{r-1}Q_v.
	\]
	For every \(j\), the orbit of \(u\) under \(\{\widetilde U_g:g\in\mathcal G\}\) contains both \(\sqrt p\,v+\sqrt{1-p}\,e_j\) and \(\sqrt p\,v-\sqrt{1-p}\,e_j\). Their sum and difference generate the directions \(v\) and \(e_j\), respectively. Hence the ranges of the \(P_g\) span \(B\), so the effective output dimension is \(r\).
	
	For \(A\), taking the expectation in \(v\) gives \(A\leq1/p\).
	Conversely, set \(s_g=1/(p|\mathcal G|)\). Then
	\[
	\sum_{g\in\mathcal G}s_gP_g
	=
	\frac{\overline P}{p}
	=
	|v\rangle\langle v|
	+
	\frac{k-p}{p(r-1)}Q_v.
	\]
	These weights are feasible exactly when \(pr\geq k\).
	For the chosen value of \(p\), this is equivalent to
	\(\sqrt r\geq\sqrt k\), which holds.
	Their total weight is \(1/p\), so
	\begin{equation}
		A=\frac1p=\frac{\sqrt r(1+\sqrt k)}{\sqrt k(1+\sqrt r)}.
		\label{eq:RankKA}
	\end{equation}
	
	We now compute \(\Sigma\).
	Let \((T,\{F_g\}_{g\in\mathcal G})\) be feasible for
	\eqref{sdp:CqSigma}. Averaging \(T\) over the group preserves
	its trace. For each \(\ell\in\mathcal G\), average
	\(\widetilde U_gF_{g^{-1}\ell}\widetilde U_g^\dagger\)
	over \(g\in\mathcal G\). The resulting operator is supported
	on \(P_\ell\), has trace \(1\), and is dominated by the
	averaged \(T\). Hence we may take
	\(
	T=\mu|v\rangle\langle v|+\nu Q_v,
	\; \mu,\nu\geq0.
	\)
	Feasibility forces \(\nu>0\). Indeed, when \(\nu=0\), the support of \(T\) is contained in \(\mathbb Cv\), whereas \(\operatorname{ran}P_g\cap\mathbb Cv=\{0\}\) because \(p<1\); hence no trace-one positive semidefinite operator supported on \(P_g\) can be dominated by \(T\).
	
	The boundary case \(\mu=0\) can occur when \(k>1\), but it is never optimal. If \(\mu=0\), then every operator \(F\) satisfying \(0\preceq F\preceq T=\nu Q_v\) is supported on \(W\). Since \(\operatorname{ran}P_g\cap W=E_g\), any feasible \(F_g\) is supported on \(E_g\), and therefore has trace at most \((k-1)\nu\). Thus feasibility requires \(\nu\geq1/(k-1)\), giving
	\(\operatorname{Tr}T\geq\frac{r-1}{k-1} \geq \frac{\sqrt r(1+\sqrt r)}{\sqrt k(1+\sqrt k)}\),
	where the second inequality is equivalent to \(\sqrt r\geq\sqrt k\). For \(k=1\), the case \(\mu=0\) is infeasible. Hence it remains only to optimize over \(\mu>0\) and \(\nu>0\).
	
	For each \(g\in\mathcal G\), write
	\(
	P_g=|u_g\rangle\langle u_g|+P_{E_g},
	\)
	where
	\(u_g=\sqrt p\,v+\sqrt{1-p}\,w_g\),
	\(E_g\subset W\), \(E_g\perp w_g\), and \(\dim E_g=k-1\). Define
	\[
	m_g(T):=
	\max\{\operatorname{Tr}F:
	0\preceq F\preceq T,\ F=P_gFP_g\}.
	\]
	Since \(T\) is invariant under unitaries acting on \(E_g\) and fixing \(E_g^\perp\), averaging an optimizer over these unitaries allows us to take
	\(
	F=cP_{E_g}+t|u_g\rangle\langle u_g|.
	\)
	We may assume \(c\leq\nu\), while the rank-one Schur complement gives
	\(t\leq[p/\mu+(1-p)/\nu]^{-1}\). Hence
	\[
	m_g(T)
	=(k-1)\nu+\frac1{p/\mu+(1-p)/\nu}.
	\]
	The constraints on \(F_g\) in \eqref{sdp:CqSigma} are feasible if and only if \(m_g(T)\geq1\). Indeed, when \(m_g(T)\geq1\), normalize an optimizer by its trace. The resulting positive semidefinite operator remains dominated by \(T\) and is supported on \(P_g\). As a density operator on that support, it is also dominated by \(P_g\).
	
	Set \(z:=\mu/\nu>0\). For fixed \(z\), minimizing
	\(\operatorname{Tr}T=\nu(z+r-1)\) subject to \(m_g(T)\geq1\)
	forces \(m_g(T)=1\). Thus
	\begin{equation}
		\nu
		=
		\frac1{k-1+1/(p/z+1-p)},
		\label{eq:RankKBx}
	\end{equation}
	and the objective becomes
	\[
	f(z)
	:=
	\frac{z+r-1}{k-1+1/(p/z+1-p)}.
	\]
	Substituting the definition of \(p\) and differentiating gives
	\[
	f'(z)=
	\frac{(\sqrt r-\sqrt k)(z-\sqrt r-1)
		[\,r+\sqrt{rk}+\sqrt k-1+z(1+\sqrt r-\sqrt k)\,]}
	{\sqrt k(1+\sqrt k)[\,(\sqrt k-1)(1+\sqrt r)+z(1+\sqrt r-\sqrt k)\,]^2}.
	\]
	Since \(1\leq k<r\) and \(z>0\), both bracketed factors are positive.
	Therefore the sign of \(f'(z)\) is determined only by \(z-\sqrt r-1\). The function is strictly decreasing on \((0,1+\sqrt r)\), strictly increasing on \((1+\sqrt r,\infty)\), and has the unique global minimizer \(z=1+\sqrt r\).
	
	Substituting this value into \eqref{eq:RankKBx} gives
	\(
	\nu=\frac1{\sqrt k(1+\sqrt k)},
	\;
	\mu=\frac{1+\sqrt r}{\sqrt k(1+\sqrt k)}.
	\)
	Hence
	\[
	\Sigma=f(1+\sqrt r)
	=\frac{\sqrt r(1+\sqrt r)}{\sqrt k(1+\sqrt k)}.
	\]
	Together with \eqref{eq:RankKA},
	\[
	\frac{\Sigma}{A}
	=
	\left(\frac{1+\sqrt r}{1+\sqrt k}\right)^2.
	\]
	
	Finally, if \(k=r\), take all output supports equal to \(I_B\). The packing constraint forces the total weight to be at most \(1\), and equality is attainable, so \(A=1\). For every feasible solution of \eqref{sdp:CqSigma}, \(F_x\preceq T\) and \(\operatorname{Tr}F_x=1\) imply \(\operatorname{Tr}T\geq1\); taking \(F_x=T=\rho\) for every \(x\), where \(\rho\) is any density operator, attains this value. Hence \(\Sigma=1\). The right-hand side also equals \(1\), proving sharpness at the boundary \(k=r\).
\end{IEEEproof}

\subsection{Operational Communication--Simulation Gap}

Classical--quantum Kraus operator spaces are closed under tensor products. Since \(\Sigma=\Sigma^{-}\) in this setting, Theorem~\ref{thm:SigmaMinusMult} gives, for every positive integer \(n\),
\(\Sigma(K^{\otimes n}) = \Sigma^{-}(K^{\otimes n}) = \Sigma^{-}(K)^n = \Sigma(K)^n\).
Thus \(S_{0,\NS}(K)=\log_2\Sigma(K)\). Together with
\(C_{0,\NS}(K)=\log_2A(K)\), this gives the following corollary.

\begin{corollary}\label{cor:OperationalGap}
	Let \(K\) be a classical--quantum Kraus operator space with effective output dimension \(r\), and suppose that every output support projection has rank at least \(k\). Then
	\[
	0\leq
	S_{0,\NS}(K)-C_{0,\NS}(K)
	\leq
	2\log_2\frac{1+\sqrt r}{1+\sqrt k}.
	\]
	For every \(1\leq k\leq r\), the upper bound is attainable.
\end{corollary}

\begin{IEEEproof}
	The lower bound is the general fact that the zero-error communication rate cannot exceed the exact simulation cost. The upper bound follows by taking logarithms in Theorem~\ref{thm:RankSensitiveSigmaA}. The construction in Proposition~\ref{prop:RankKFamily} attains \(\Sigma/A=c_{r,k}^2\) for every \((r,k)\), and hence also attains the upper bound on the operational gap.
\end{IEEEproof}

\appendices

\section{Duality and Strong Duality for the Primal Formulation of \texorpdfstring{$\Sigma^{-}$}{Sigma-minus}}\label{app:SigmaMinusDual}
Recall the equivalent primal formulation \eqref{sdp:SigmaMinusPrimal} of \(\Sigma^{-}\):
\[
\begin{aligned}
	\Sigma^{-}(K)=\min\quad &\operatorname{Tr} T_B\\
	\mathrm{s.t.}\quad
	&V_{AB}\succeq0,\quad V_{AB}=PV_{AB}P,\\
	&\operatorname{Tr}_BV_{AB}\succeq I_A,\\
	&V_{AB}\preceq I_A\otimes T_B,\quad T_B\succeq0.
\end{aligned}
\]
Define
\(\mathcal C_P:=\{V_{AB}:V_{AB}\succeq0,\ V_{AB}=P_{AB}V_{AB}P_{AB}\}\).
The SDP can then be written as a minimization problem with objective \(\operatorname{Tr} T_B\) and constraints
\(V_{AB}\in\mathcal C_P,\; T_B\succeq0,\; \operatorname{Tr}_BV_{AB}-I_A\succeq0,\; I_A\otimes T_B-V_{AB}\succeq0\).
Introduce positive semidefinite multipliers \(R_A\) and \(X_{AB}\) for the last two constraints, respectively. The Lagrangian is
\begin{align*}
	\mathcal L(V,T;R,X)
	={}&\operatorname{Tr} T_B
	-\operatorname{Tr} R_A(\operatorname{Tr}_BV_{AB}-I_A)-\operatorname{Tr} X_{AB}(I_A\otimes T_B-V_{AB})\\
	={}&\operatorname{Tr} R_A
	+\operatorname{Tr}\!\left[(I_B-\operatorname{Tr}_AX_{AB})T_B\right]+\operatorname{Tr}\!\left[(X_{AB}-R_A\otimes I_B)V_{AB}\right].
\end{align*}
The infimum of \(\mathcal L\) over \(T_B\succeq0\) is finite if and only if \(\operatorname{Tr}_AX_{AB}\preceq I_B\). Its infimum over \(V\in\mathcal C_P\) is finite if and only if
\[
(X_{AB}-R_A\otimes I_B)\in \mathcal C_P^*
=
\{Y_{AB}:P_{AB}Y_{AB}P_{AB}\succeq0\}
\iff
P_{AB}(X_{AB}-R_A\otimes I_B)P_{AB}\succeq0.
\]
The remaining constant term is \(\operatorname{Tr} R_A\), which yields \eqref{sdp:SigmaMinusDual}.

To prove strong duality for the primal SDP of \(\Sigma^{-}(K)\), note that since \(K\) is a Kraus operator space, there exists a CPTP channel with Choi operator \(J_{AB}\) satisfying
\(
J_{AB}\succeq0,
\;
\supp(J_{AB})=P_{AB},
\;
\operatorname{Tr}_BJ_{AB}=I_A.
\)
Choose \(c>1\) and set \(V_{AB}=cJ_{AB}\). Then \(V_{AB}\) is strictly positive definite on \(\operatorname{ran}P_{AB}\), and hence lies in the relative interior of \(\mathcal C_P\), while \(\operatorname{Tr}_BV_{AB}-I_A=(c-1)I_A\succ0\). Taking \(T_B=\lambda I_B\) for sufficiently large \(\lambda\) gives \(I_A\otimes T_B-V_{AB}\succ0\). Thus the equivalent primal formulation satisfies the relative Slater condition. The finite-dimensional Slater theorem \cite[Theorem~2.2]{Watrous2009SDP} therefore proves strong duality for \eqref{sdp:SigmaMinusPrimal}.

\section{Duality and Strong Duality for \texorpdfstring{$\Gamma$}{Gamma}}\label{app:GammaDual}

In this section, we use the real Hilbert--Schmidt inner product
\(\langle R,S\rangle_{\mathbb R}:=\operatorname{Re}\operatorname{Tr}(R^{\dagger}S)\) on the space of complex matrices.
Recall the equivalent SDP formulation \eqref{sdp:GammaEq} of \(\Gamma\). We use this formulation to derive the dual \eqref{sdp:GammaDual}.

\begin{equation}
	\begin{aligned}
		\Gamma(K)=\min\quad &\frac12\operatorname{Tr}(T_{0,B}+T_{1,B})\\
		\mathrm{s.t.}\quad
		&PX_{AB}=X_{AB},\\
		&\operatorname{Tr}_BX_{AB}= I_A,\\
		&\begin{pmatrix}
			I_A\otimes T_{0,B}&X_{AB}\\
			X_{AB}^{\dagger}&I_A\otimes T_{1,B}
		\end{pmatrix}\succeq0.
	\end{aligned}
	\label{sdp:GammaEqAppendix}
\end{equation}

Introduce a block slack variable \(M\succeq0\), and use the shorthand
\(\Phi(T_{0,B},T_{1,B},X_{AB}) := \begin{psmallmatrix} I_A\otimes T_{0,B}&X_{AB}\\ X_{AB}^{\dagger}&I_A\otimes T_{1,B} \end{psmallmatrix}\).
Then \eqref{sdp:GammaEq} can be rewritten equivalently as
\begin{equation*}
	\begin{aligned}
		\Gamma(K)=\min\quad&
		\frac12\operatorname{Tr}(T_{0,B}+T_{1,B})\\
		\mathrm{s.t.}\quad&
		Q_{AB}X_{AB}=0,\\
		&\operatorname{Tr}_BX_{AB}=I_A,\\
		&M=\Phi(T_0,T_1,X),\\
		&M\succeq0.
	\end{aligned}
\end{equation*}
For the equality \(M=\Phi(T_0,T_1,X)\), introduce the Hermitian multiplier
\(\Lambda = \frac12 \begin{psmallmatrix} Y_{0,AB}&-Z_{AB}\\ -Z_{AB}^{\dagger}&Y_{1,AB} \end{psmallmatrix}\).
For the complex-linear equalities \(Q_{AB}X_{AB}=0\) and
\(\operatorname{Tr}_BX_{AB}=I_A\), introduce arbitrary complex multipliers
\(W_{AB}\) and \(H_A\), respectively. The corresponding Lagrangian is
\begin{align*}
	\mathcal L=
	\frac12\operatorname{Tr}(T_{0,B}+T_{1,B})
	+\langle\Lambda,M-\Phi(T_0,T_1,X)\rangle_{\mathbb R}
	+\langle W_{AB},Q_{AB}X_{AB}\rangle_{\mathbb R}
	+\langle H_A,I_A-\operatorname{Tr}_BX_{AB}\rangle_{\mathbb R}.
\end{align*}
Expanding this expression gives
\begin{align*}
	\mathcal L
	={}&
	\operatorname{Re}\operatorname{Tr} H_A
	+\langle\Lambda,M\rangle_{\mathbb R}
	\nonumber\\
	&+
	\frac12
	\operatorname{Tr}\!\left[
	\bigl(I_B-\operatorname{Tr}_A Y_{0,AB}\bigr)T_{0,B}
	\right]
	\nonumber\\
	&+
	\frac12
	\operatorname{Tr}\!\left[
	\bigl(I_B-\operatorname{Tr}_A Y_{1,AB}\bigr)T_{1,B}
	\right]
	\nonumber\\
	&+
	\left\langle
	Z_{AB}+Q_{AB}W_{AB}-H_A\otimes I_B,
	X_{AB}
	\right\rangle_{\mathbb R}.
\end{align*}

The infimum of \(\mathcal L\) over \(M\succeq0\) is finite if and only if \(\Lambda\succeq0\), namely,
\(
\begin{psmallmatrix}
	Y_{0,AB}&-Z_{AB}\\
	-Z_{AB}^{\dagger}&Y_{1,AB}
\end{psmallmatrix}
\succeq0.
\)
This is equivalent to
\(
\begin{psmallmatrix}
	Y_{0,AB}&Z_{AB}\\
	Z_{AB}^{\dagger}&Y_{1,AB}
\end{psmallmatrix}
\succeq0.
\)

The infimum over \(T_{0,B}\) and \(T_{1,B}\) is finite if and only if
\(\operatorname{Tr}_A Y_{0,AB}=I_B,\; \operatorname{Tr}_A Y_{1,AB}=I_B\).
For fixed multipliers, the infimum over \(X_{AB}\) is finite 
if and only if
\(Z_{AB}+Q_{AB}W_{AB} = H_A\otimes I_B\).
Left-multiplying by \(P_{AB}\) and using \(P_{AB}Q_{AB}=0\) gives the necessary condition
\(P_{AB}Z_{AB} = P_{AB}(H_A\otimes I_B)\).
Conversely, if this equality holds, then
\(
H_A\otimes I_B-Z_{AB}
=
Q_{AB}(H_A\otimes I_B-Z_{AB}),
\)
so choosing \(W_{AB}:=H_A\otimes I_B-Z_{AB}\) satisfies the original condition. Hence such a \(W_{AB}\) exists if and only if
\(P_{AB}Z_{AB} = P_{AB}(H_A\otimes I_B)\).

Define the corresponding dual function by
\[
g(\Lambda,W,H)
:=
\inf_{\substack{
		M\succeq0,\\
		T_{0,B},T_{1,B}\in\operatorname{Herm}(B),\\
		X_{AB}\in\mathcal L(A\otimes B)
}}
\mathcal L(M,T_0,T_1,X;\Lambda,W,H),
\]
where \(\Lambda\) is Hermitian and \(W_{AB}\) and \(H_A\) are arbitrary complex operators. Taking the supremum of the dual function over these multipliers yields the dual program below.

The preceding calculation gives
\begin{align*}
	g(\Lambda,W,H)
	=
	\begin{cases}
		\operatorname{Re}\operatorname{Tr} H_A,
		&
		\begin{aligned}
			&\Lambda\succeq0,\\[-1mm]
			&\operatorname{Tr}_A Y_{0,AB}=\operatorname{Tr}_A Y_{1,AB}=I_B,\\[-1mm]
			&Z_{AB}+Q_{AB}W_{AB}=H_A\otimes I_B,
		\end{aligned}
		\\[6mm]
		-\infty,
		&\text{otherwise}.
	\end{cases}
\end{align*}
Therefore, the dual of \(\Gamma\) is
\begin{equation*}
	\begin{aligned}
		\max\quad&
		\operatorname{Re}\operatorname{Tr} H_A\\
		\mathrm{s.t.}\quad&
		\begin{pmatrix}
			Y_{0,AB}&Z_{AB}\\
			Z_{AB}^{\dagger}&Y_{1,AB}
		\end{pmatrix}
		\succeq0,\\
		&\operatorname{Tr}_A Y_{0,AB}=I_B,\\
		&\operatorname{Tr}_A Y_{1,AB}=I_B,\\
		&P_{AB}Z_{AB}
		=
		P_{AB}(H_A\otimes I_B),
	\end{aligned}
\end{equation*}
where \(H_A\) and \(Z_{AB}\) are arbitrary complex operators, and
\(Y_{0,AB}\) and \(Y_{1,AB}\) are Hermitian operators.
This agrees with the dual formulation \eqref{sdp:GammaDual} in the main text.

We now prove strong duality. Since \(K\) is a Kraus operator space, there exists a CPTP channel with Choi operator \(J_{AB}\succeq0\) such that
\(
P_{AB}J_{AB}=J_{AB},
\;
\operatorname{Tr}_BJ_{AB}=I_A.
\)
In \eqref{sdp:GammaEqAppendix}, take
\(
X_{AB}=J_{AB},
\;
T_{0,B}=T_{1,B}=tI_B,
\)
where \(t\) is sufficiently large that
\(
	\begin{psmallmatrix}
		tI_{AB}&J_{AB}\\
		J_{AB}&tI_{AB}
	\end{psmallmatrix}
	\succ0.
\)
This gives the strictly feasible solution
\(X_{AB}=J_{AB}\), \(T_{0,B}=tI_B\), and \(T_{1,B}=tI_B\).
Therefore, the finite-dimensional Slater strong-duality theorem for semidefinite programs \cite[Theorem~2.2]{Watrous2009SDP} shows that the optimal value of the dual program equals \(\Gamma(K)\).

\section{Classical--Quantum Reductions and Dual Programs}\label{app:CqReductions}

\subsection{Pinching and the Identity \texorpdfstring{$\Sigma^{-}=\Sigma$}{Sigma-minus = Sigma}}

Define the pinching map on the classical register by
\(
\Delta_A(H)
:=
\sum_x |x\rangle\langle x|_A H |x\rangle\langle x|_A,
\;
\Delta_{AB}:=\Delta_A\otimes\operatorname{id}_B.
\)
For a classical--quantum Kraus operator space,
\(P_{AB}=\sum_x|x\rangle\langle x|_A\otimes P_x\)
and \(Q_{AB}=I_{AB}-P_{AB}\) are block diagonal in the classical basis. Moreover,
\(\Delta_{AB}(I_A\otimes T_B)=I_A\otimes T_B\). Thus pinching preserves the positivity, support, partial-trace, and domination constraints in the relevant primal SDPs without changing their objective values.

First consider \(\Sigma^{-}\) in \eqref{sdp:SigmaMinusPrimal}. For a feasible pair \((V_{AB},T_B)\), write
\(\Delta_{AB}(V_{AB})=\sum_x |x\rangle\langle x|_A\otimes V_x\).

Then, for every \(x\),
\(V_x\succeq0,\; V_x=P_xV_xP_x,\; \operatorname{Tr}V_x\geq1,\; V_x\preceq T_B\).
Set \(c_x:=\operatorname{Tr}V_x\) and \(F_x:=V_x/c_x\). Since \(c_x\geq1\),
\(F_x\succeq0,\; F_x=P_xF_xP_x,\; \operatorname{Tr}F_x=1,\; F_x\preceq V_x\preceq T_B\).
Because \(F_x\) is a density operator supported on \(\operatorname{ran}P_x\), one also has \(F_x\preceq P_x\). Hence every feasible solution of \(\Sigma^{-}\) gives a feasible solution of \eqref{sdp:CqSigma} with the same objective value.

Conversely, suppose that \((T_B,\{F_x\}_x)\) is feasible for \eqref{sdp:CqSigma}, and define
\(F_{AB}:=\sum_x |x\rangle\langle x|_A\otimes F_x\).
Then
\(F_{AB}\succeq0,\; F_{AB}=P_{AB}F_{AB}P_{AB},\; \operatorname{Tr}_BF_{AB}=I_A,\; F_{AB}\preceq I_A\otimes T_B\).
Thus \((F_{AB},T_B)\) is feasible for \eqref{sdp:SigmaMinusPrimal} with the same objective value. This proves that \(\Sigma^{-}(K)\) equals the first reduced program in Proposition~\ref{prop:CqReductions}.

The same pinching argument applied to the semidefinite covering number \(\Sigma(K)\) gives exactly the same reduced program. In that case the block normalization \(\operatorname{Tr}F_x=1\) is already part of the primal equality constraint. Therefore
\[
\Sigma^{-}(K)=\Sigma(K)
\]
for every classical--quantum Kraus operator space.

\subsection{The Classical--Quantum Primal and Dual of \texorpdfstring{$\Gamma$}{Gamma}}

Consider the normalized primal formulation \eqref{sdp:GammaEq}. If \((X_{AB},T_0,T_1)\) is feasible, complete positivity of \(\operatorname{id}_2\otimes\Delta_{AB}\) gives
\[
\begin{pmatrix}
	I_A\otimes T_0&\Delta_{AB}(X_{AB})\\
	\Delta_{AB}(X_{AB})^\dagger&I_A\otimes T_1
\end{pmatrix}
\succeq0.
\]
Write \(\Delta_{AB}(X_{AB})=\sum_x |x\rangle\langle x|_A\otimes X_x\).
The conditions \(P_{AB}X_{AB}=X_{AB}\) and \(\operatorname{Tr}_BX_{AB}=I_A\) imply
\(
P_xX_x=X_x,
\;
\operatorname{Tr}X_x=1.
\)
The block positivity condition is the direct sum of
\(
\begin{psmallmatrix}
	T_0&X_x\\
	X_x^\dagger&T_1
\end{psmallmatrix}\succeq0
\)
over \(x\). This proves one direction of \eqref{sdp:CqGammaPrimal}. Conversely, any family satisfying these reduced constraints can be assembled as
\(X_{AB}=\sum_x|x\rangle\langle x|\otimes X_x\), which is feasible for \eqref{sdp:GammaEq} with the same objective value. Hence \eqref{sdp:CqGammaPrimal} is equivalent to the general primal.

For the dual, start from \eqref{sdp:GammaDual}. Given a feasible quadruple \((H_A,Y_0,Y_1,Z)\), replace it by
\(
\widetilde H_A:=\Delta_A(H_A),
\;
\widetilde Y_j:=\Delta_{AB}(Y_j),
\;
\widetilde Z:=\Delta_{AB}(Z).
\)
The objective is unchanged. Complete positivity preserves block positivity, the partial-trace equalities are preserved, and block diagonality of \(P_{AB}\) gives
\(P_{AB}\widetilde Z = P_{AB}(\widetilde H_A\otimes I_B)\).
Thus an optimal solution can be chosen block diagonal on the classical register. Write
\(\widetilde H_A=\sum_xh_x|x\rangle\langle x|\),
\(\widetilde Y_j=\sum_x|x\rangle\langle x|\otimes Y_{j,x}\), and
\(\widetilde Z=\sum_x|x\rangle\langle x|\otimes Z_x\).
The constraints become
\[
\begin{pmatrix}Y_{0,x}&Z_x\\Z_x^\dagger&Y_{1,x}\end{pmatrix}\succeq0,
\quad
\sum_xY_{0,x}=\sum_xY_{1,x}=I_B,
\quad
P_xZ_x=h_xP_x.
\]
If \(h_x=|h_x|e^{i\phi_x}\), replace \(Z_x\) by \(e^{-i\phi_x}Z_x\) and \(h_x\) by \(|h_x|\). The corresponding \(2\times2\) block is obtained by unitary conjugation, so positivity is preserved, while the objective changes from \(\operatorname{Re}h_x\) to \(|h_x|\). Hence an optimal solution may be chosen with \(h_x\geq0\), yielding \eqref{sdp:CqGammaDual}.

The two normalization equalities in \eqref{sdp:CqGammaDual} can equivalently be replaced by
\begin{equation}
	\sum_xY_{0,x}\preceq I_B,
	\quad
	\sum_xY_{1,x}\preceq I_B.
	\label{eq:CqGammaDualIneqNorm}
\end{equation}
Indeed, every solution satisfying the two equalities also satisfies \eqref{eq:CqGammaDualIneqNorm}. Conversely, for a feasible solution satisfying only \eqref{eq:CqGammaDualIneqNorm}, define
\(D_j:=I_B-\sum_xY_{j,x}\succeq0\) for \(j=0,1\).
Fix an arbitrary input \(x_0\) and replace \(Y_{j,x_0}\) by \(Y_{j,x_0}+D_j\). The affected block matrix is increased by \(\operatorname{diag}(D_0,D_1)\succeq0\); all \(Z_x\), \(h_x\), support constraints, and the objective remain unchanged. The normalization inequalities can therefore be completed to equalities without changing the optimal value.

Strong duality for the reduced \(\Gamma\) programs follows from strong duality of the general programs and the value-preserving pinching reductions above.

\subsection{The Classical--Quantum Dual of \texorpdfstring{$\Sigma$}{Sigma}}

Starting from \eqref{sdp:CqSigma}, the condition \(F_x\preceq P_x\) is redundant once
\(F_x\succeq0\), \(F_x=P_xF_xP_x\), and \(\operatorname{Tr}F_x=1\) are imposed. Introduce \(Y_x\succeq0\) for the constraints \(T_B-F_x\succeq0\) and real multipliers \(\lambda_x\) for \(\operatorname{Tr}F_x=1\). The Lagrangian is
\[
\mathcal L
=
\operatorname{Tr}\!\left[\left(I_B-\sum_xY_x\right)T_B\right]
+
\sum_x\operatorname{Tr}\!\left[(Y_x-\lambda_x I_B)F_x\right]
+
\sum_x\lambda_x.
\]
The infimum over \(T_B\succeq0\) is finite exactly when
\(\sum_xY_x\preceq I_B\). The infimum over positive semidefinite \(F_x\) supported on \(P_x\) is finite exactly when
\[
P_xY_xP_x\succeq\lambda_xP_x.
\]
If \(\lambda_x<0\), replacing it by \(0\) preserves feasibility and increases the objective, so an optimal dual solution may be chosen with \(\lambda_x\geq0\). This yields \eqref{sdp:CqSigmaDual}.

For strong duality, use the equivalent primal formulation in which the redundant constraints \(F_x\preceq P_x\) have been removed. Take
\(F_x=P_x/\operatorname{rank}P_x\) and \(T_B=tI_B\) with \(t\) sufficiently large. This is strictly feasible relative to the fixed support faces containing the variables \(F_x\). Hence the finite-dimensional relative Slater theorem gives strong duality for \eqref{sdp:CqSigma} and \eqref{sdp:CqSigmaDual}.
\section{Supplementary Comparison with the Quantum Lov\'asz Number}\label{app:LovaszComparison}

\begin{proposition}
	\label{prop:GammaSigmaMinusVsQuantumLovasz}
	For every Kraus operator space \(K\),
	\[
	\Theta_*(S(K))
	\leq
	\Theta_{\NS}(K)
	\leq
	\Gamma(K)
	\leq
	\Sigma^{-}(K).
	\]
	Moreover, when viewed as functions of \(K\), both \(\Gamma(K)\) and \(\Sigma^{-}(K)\) are incomparable with the quantum Lov\'asz number
	\(\widetilde\vartheta(S(K))\) \cite{DuanSeveriniWinter2013}.
\end{proposition}
\begin{IEEEproof}
	The first inequality is the known relation \eqref{eq:ThetaStarThetaNS}; the remaining inequalities are Corollary~\ref{cor:ThetaGamma} and Proposition~\ref{prop:GammaSigmaMinus}.
	We next exhibit examples yielding strict inequalities in both directions relative to
	\(\widetilde\vartheta\).
	
	First, consider the qutrit--qutrit channel family
	\(\mathcal N_\alpha\) constructed by Wang and Duan, where \(0<\alpha\leq\pi/4\), and let \(K_\alpha:=K(\mathcal N_\alpha)\). Proposition~1 and Eqs.~(16)--(18) of \cite{WangDuan2018} give \(C_{0,\NS}(K_\alpha)=2\) and \(\Sigma(\mathcal N_\alpha)\leq4\).
	Duan and Winter showed that the semidefinite program \eqref{sdp:Sigma} is equivalent to the following channel formulation \cite{DuanWinter2016}:
	\(\Sigma(K) = \min\left\{ \Sigma(\mathcal M): \mathcal M\ \text{is a quantum channel and}\ K(\mathcal M)\subseteq K \right\}\),
	where \(\Sigma(\mathcal M)\) denotes the SDP value associated with one-shot QNS-assisted exact simulation of the fixed channel
	\(\mathcal M\).
	It follows that
	\(\Sigma(K_\alpha) \leq \Sigma(\mathcal N_\alpha) \leq4\).
	On the other hand,
	\(\Theta_{\NS}(K_\alpha) = 2^{C_{0,\NS}(K_\alpha)} = 4\).
	Together with Theorem~\ref{thm:MainHierarchy}, the inequalities \(\Sigma(K_\alpha)\leq4\) and \(\Theta_{\NS}(K_\alpha)=4\) imply
	\(\Gamma(K_\alpha) = \Sigma^{-}(K_\alpha) = \Sigma(K_\alpha) = 4\).
	Wang and Duan computed in Proposition~2 and Eq.~(20) of \cite{WangDuan2018} that
	\(\widetilde\vartheta(S(K_\alpha)) = 2+\cos^2\alpha+\cos^{-2}\alpha > 4\).
	Therefore,
	\begin{equation}
		\Gamma(K_\alpha)
		=
		\Sigma^{-}(K_\alpha)
		<
		\widetilde\vartheta(S(K_\alpha)).
		\label{eq:GammaSigmaMinusBelowQuantumLovasz}
	\end{equation}
	
	We now consider a classical channel with input alphabet \(X\) and output alphabet \(Y\). Let \(E_K\subseteq X\times Y\) be the set of input--output pairs with nonzero transition probability. Its Kraus operator space can be written as
	\(K = \operatorname{span} \left\{ |y\rangle_B\langle x|_A: (x,y)\in E_K \right\}\).
	Define the input--output support bipartite graph of \(K\) by \(\mathsf H_K:=(X,Y,E_K)\).
	Its confusability graph \(G\) has vertex set \(X\), with two distinct inputs \(x,x'\in X\) adjacent if and only if they have a common output symbol, namely,
	\(
	x\sim_G x'
	\iff 
	\exists\,y\in Y
	\)
	such that
	\(
	(x,y),(x',y)\in E_K
	\).
	
	The fractional packing number of the bipartite graph \(\mathsf H_K\) is defined by
	\[
	\alpha^*(\mathsf H_K)
	:=
	\max
	\left\{
	\sum_{x\in X}p_x:
	p_x\geq0,\;
	\sum_{\substack{x\in X\\(x,y)\in E_K}}p_x\leq1
	\ \text{for every }y\in Y
	\right\}.
	\]
	Equation~(24) in Section~IV of \cite{DuanWinter2016}, together with the discussion immediately following it, gives \(\Sigma(K)=A(K)=\alpha^*(\mathsf H_K)\).
	
	Let \(X=Y=\mathbb Z_5\), and consider the classical Kraus operator space
	\[
	K_5
	:=
	\operatorname{span}
	\left\{
	|x\rangle_B\langle x|_A,\,
	|x\rangle_B\langle x+1|_A:
	x\in\mathbb Z_5
	\right\},
	\]
	where all indices are taken modulo \(5\). For each output \(x\), the inputs that can produce that output are precisely \(x\) and \(x+1\). Hence the confusability graph of \(\mathsf H_{K_5}\) is the pentagon \(C_5\), and
	\[
	\alpha^*(\mathsf H_{K_5})
	=
	\max
	\left\{
	\sum_{x\in\mathbb Z_5}p_x:
	p_x+p_{x+1}\leq1,\;
	p_x\geq0
	\right\}
	=
	\frac52.
	\]
	Indeed, summing the five constraints gives \(2\sum_{x\in\mathbb Z_5}p_x\leq5\), and the upper bound is attained by taking \(p_x=1/2\). Therefore, \(A(K_5)=\Sigma(K_5)=\frac52\). Theorem~\ref{thm:MainHierarchy} together with Proposition~\ref{prop:CqReductions} gives
	\(A(K)\leq\Gamma(K)\leq\Sigma^{-}(K)=\Sigma(K)\), and hence
	\(\Gamma(K_5) = \Sigma^{-}(K_5) = \Sigma(K_5) = \frac52\).
	
	On the other hand, using the definition \(S(K):=K^\dagger K\), a direct calculation from the expression for \(K_5\) gives
	\[
	S(K_5)
	=
	\operatorname{span}
	\left\{
	|x\rangle\langle x|,\,
	|x\rangle\langle x+1|,\,
	|x+1\rangle\langle x|:
	x\in\mathbb Z_5
	\right\}
	=
	S_{C_5}.
	\]
	Duan, Severini, and Winter proved in Section~IV of \cite{DuanSeveriniWinter2013} that the quantum Lov\'asz number reduces to the classical Lov\'asz number on operator systems associated with classical graphs \cite{Lovasz1979}. Consequently,
	\(\widetilde\vartheta(S(K_5)) = \widetilde\vartheta(S_{C_5}) = \vartheta(C_5) = \sqrt5 < \frac52\).
	Thus,
	\begin{equation}
		\widetilde\vartheta(S(K_5))
		<
		\Gamma(K_5)
		=
		\Sigma^{-}(K_5).
		\label{eq:GammaSigmaMinusAboveQuantumLovasz}
	\end{equation}
	
	Equations~\eqref{eq:GammaSigmaMinusBelowQuantumLovasz} and \eqref{eq:GammaSigmaMinusAboveQuantumLovasz} establish strict inequalities in opposite directions. Therefore, when viewed as functions of \(K\), both \(\Gamma(K)\) and \(\Sigma^{-}(K)\) are incomparable with the quantum Lov\'asz number
	\(\widetilde\vartheta(S(K))\).
\end{IEEEproof}

\section*{Acknowledgment}

During the course of this research, the authors used ChatGPT (OpenAI) in developing the mathematical arguments and proofs in IV-E and V-B--V-G. The relevant material was reviewed and verified by the authors through multiple rounds of manual checking. The authors take full responsibility for the content of this paper.

\bibliographystyle{IEEEtran}
\bibliography{references}

\end{document}